\documentclass[opre,nonblindrev]{informs3} %
\DoubleSpacedXI %

\usepackage[ruled]{algorithm2e}

\SetKwInput{KwPara}{Parameters}
\SetAlFnt{\small}
\SetAlCapFnt{\small}
\SetAlCapNameFnt{\small}
\SetAlCapHSkip{0pt}

\newcommand{\prob}{\ensuremath{\mathbf{P}}}
\newcommand{\expec}{\ensuremath{\mathbf{E}}}
\newcommand{\ind}{\ensuremath{\mathbf{I}}}
\newcommand{\reals}{\ensuremath{\mathbb{R}}}

\newcommand{\defeq}{\ensuremath{\triangleq}}

\usepackage{graphics, graphicx} %
\usepackage{subcaption}

\usepackage[showdeletions]{color-edits}

\newcommand{\instance}{\ensuremath{\mathcal{I}}\xspace}
\newcommand{\alg}{\ensuremath{\mathfrak{a}}\xspace}
\newcommand{\full}{\ensuremath{\mathfrak{Full}}\xspace}
\newcommand{\opt}{\ensuremath{\mathsf{OPT}}}
\newcommand{\regret}{\ensuremath{\mathsf{Reg}}}

\newcommand{\gap}{\ensuremath{\mathsf{Gap}}}
\newcommand{\pers}{\ensuremath{\mathsf{Pers}}}
\newcommand{\relint}{\ensuremath{\mathsf{relint}}}
\newcommand{\event}{\ensuremath{\mathcal{E}_T}}

\newcommand{\kl}[2]{\ensuremath{\mathsf{KL}\left({#1||#2}\right)}}
\newcommand{\leap}{\ensuremath{\mathfrak{Rai}}\xspace}

\usepackage{amsmath, amssymb} %
\usepackage[all,warning]{onlyamsmath}

\usepackage[margin=1in]{geometry} %
\usepackage{booktabs} %
\usepackage{units} %
\usepackage{microtype} %
\usepackage{enumitem} %
\usepackage{multicol}

\usepackage{fancyhdr} %
\usepackage{url} %
\usepackage[normalem]{ulem}
\usepackage[dvipsnames,svgnames]{xcolor} %

\usepackage{xspace}

\usepackage{tikz} %
\usetikzlibrary{shapes,decorations}
\usetikzlibrary{arrows}
\usepackage{pgfplots} %
\pgfplotsset{compat=newest} %

\usepackage{comment}

\usepackage{natbib}
 \bibpunct[, ]{(}{)}{,}{a}{}{,}%
 \def\bibfont{\small}%
\TheoremsNumberedThrough     %
\ECRepeatTheorems

\EquationsNumberedThrough    %
\begin{document}

\RUNAUTHOR{Zu, Iyer, and Xu}
\RUNTITLE{Learning to Persuade on the Fly}

\TITLE{Learning to Persuade on the Fly: Robustness Against Ignorance}

\ARTICLEAUTHORS{%
  \AUTHOR{You Zu} \AFF{Industrial and Systems Engineering, University
    of Minnesota, Minneapolis, MN, 55455,
    \EMAIL{zu000002@umn.edu}} %
  \AUTHOR{Krishnamurthy Iyer} \AFF{Industrial and Systems Engineering,
    University of Minnesota, Minneapolis, MN 55455,
    \EMAIL{kriyer@umn.edu}}%
  \AUTHOR{Haifeng Xu} \AFF{Department of Computer Science, University
    of Chicago, Chicago, IL  60637,
    \EMAIL{haifengxu@uchicago.edu}}
} %

\ABSTRACT{%
%
Motivated by information sharing in online platforms, we study
repeated persuasion between a sender and a stream of receivers where
at each time, the sender observes a payoff-relevant state drawn
independently and identically from an \emph{unknown} distribution, and
shares state information with the receivers who each choose an action.
The sender seeks to persuade the receivers into taking actions aligned
with the sender's preference by selectively sharing state information.
However, in contrast to the standard models, neither the sender nor
the receivers know the distribution, and the sender has to persuade
while learning the distribution on the fly.

We study the sender's learning problem of making \emph{persuasive}
action recommendations to achieve low regret against the optimal
persuasion mechanism with the knowledge of the distribution. To do
this, we first propose and motivate a persuasiveness criterion for the
unknown distribution setting that centers {\em robustness} as a
requirement in the face of uncertainty. Our main result is an
algorithm that, with high probability, is robustly-persuasive and
achieves $O(\sqrt{T\log T})$ regret, where $T$ is the horizon length.
Intuitively, at each time our algorithm maintains a set of candidate
distributions, and chooses a signaling mechanism that is
simultaneously persuasive for all of them. Core to our proof is a
tight analysis about the cost of robust persuasion, which may be of
independent interest. We further prove that this regret order is
optimal (up to logarithmic terms) by showing that no algorithm can
achieve regret better than $\Omega(\sqrt{T})$.

  %

  %
  %
  %
  %
  %

  %
  %
  %
  %

  %
  %
  %
  %
  %
  %
  %
  %
  %
  %
  %
  %
  %
  %
  %

  %
  %
  %
  %
  %
  %
  %
  %
  %
  %
  %
  %
  %
  %
  %
  %
  %
  %
  %

%
%
%
%
%
%
 %
}

\KEYWORDS{no-regret learning, robustness, persuasion, prior-independence}

\maketitle

%
\section{Introduction}
\label{sec:introduction}
Examples of online platforms recommending content or products
  to their users abound in online economy. For instance, online
  retailers like Amazon or Etsy recommend products from third-party
  sellers to users, styling services like Stitch Fix recommend
  clothing designs made by custom brands, and online platforms like
  YouTube or Spotify recommend content or playlist generated by
  creators. There are two intrinsic challenges in such online
  recommendations, which we address simultaneously in this paper.
  First, the platform making such recommendations often needs to
  balance the dual objectives of being persuasive (i.e., making
  obedient recommendations that will be adopted by the
  users~\citep{bergemann2016bayes}) as well as furthering the
  platform's goals such as increased sales, fewer returns or more
  engaged users. Second, the platform often faces a large volume of
  \emph{new} products/contents/services with {\em a priori} unknown
  quality/reward distributions and thus has to learn to make good
  recommendation. We tackle these two challenges by studying learning
  to persuade on the fly.

\subsection{Motivating Applications}
To motivate the problem we consider, we now describe two concrete
examples in the domain of two-sided platforms.

\begin{example}[\textbf{Content recommendations by online media
    platforms}] Consider a media platform like YouTube or TikTok, that
  recommends content created by independent creators (``channels'') to
  its users. New channels regularly join the platform, and start
  producing content whose {\em quality distribution\/} is unknown to
  both the platform and its users. Here, by a content's quality, we
  refer to how engaging, interesting or relevant the users find the
  content. Despite this lack of knowledge, the platform faces the
  problem of deciding whether to recommend content from such new
  channels to its stream of users. In this context, the users seek to
  consume fresh and high-quality content, while the platform itself
  may have other goals, such as maximizing user engagement or
  increasing channel exposure, which are not fully aligned with users'
  interests. Furthermore, from extensive user-level data, the platform
  may have good estimates about the utility a user derives from
  consuming a particular content. A user encountering a new channel
  may have a prior belief about its quality distribution based on
  their past experiences in the platform, and from any information
  provided by the channel itself on their profile. Furthermore, the
  user may have additional (partial) information from any reviews or
  ratings left by previous users (or similar summary statistics). For
  each new content from a channel, the platform observes its quality
  (perhaps after an initial exploration or through in-house reviewers)
  and decides whether or not to recommend the content to its users. If
  the platform and the users know a channel's content quality
  distribution, the platform can reliably make recommendations that
  optimize its own goals while maintaining user satisfaction, by
  consistently mixing high-quality content with some mediocre ones.
  However, given the lack of such distributional information, the
  platform must \emph{learn} to make such recommendations over time,
  as the channel produces more content.~\Halmos
\end{example}

\begin{example}[\textbf{Recommendations on hiring platforms}]
  Consider a hiring platform, where employers receive recommendations
  about candidates for recruitment (e.g., ``recommended matches'' in
  LinkedIn Recruiter). These recommendations are typically tailored to
  the employer's project requirements. However, within the set of
  candidates satisfying the requirements, there would be a range of
  capabilities/fit, whose {\em distribution\/} would be unknown to the
  platform or the employer. Nevertheless, for any particular candidate
  who might be interested in the position, the platform may be able to
  assess the candidate's capability based on various candidate
  features, such as her endorsements, references, etc, using which the
  platform decides whether or not to recommend the candidate.
  Similarly, the employer through the course of interviewing different
  candidates may learn about the capability distribution. While the
  employer would prefer to be matched with few high-capability
  candidates to interview, the platform may have additional incentives
  from having to cater to the candidates-side of the market, such as
  increasing the overall number of interviews. Once again, if the
  distribution of the candidates' capabilities is known to the
  platform and the employer, the platform could reliably recommend
  candidates to optimize its goals while simultaneously meeting the
  employer's preferences. But, without such information, the platform
  needs to learn to recommend candidates as they apply over
  time.~\Halmos
\end{example}

This paper studies the problem faced by such a platform learning to
make {\em persuasive} recommendations to a stream of users. While
previous work has studied information design in two-sided markets ---
ranging from recommending products from third-party sellers on
e-commerce platforms like Amazon and eBay~\citep{gur2022information,
  elliott2022matching}, recommending drivers by sharing demand trend
on ride-sharing services like Uber and
Lyft~\citep{yang2019information}, to accommodation and rental
recommendations in Airbnb~\citep{ romanyuk2019cream} --- the common
assumption is that the platform knows the underlying state
distribution. Our work contributes to this literature by relaxing this
strong assumption.

\subsection{Modeling contributions}

Formally, we study a repeated persuasion setting between a sender and
a stream of receivers, where at each time $t$, the sender shares some
information correlated to some payoff-relevant \emph{state} with the
corresponding receiver. The state at each time $t$ is drawn
independently and identically from an unknown distribution, and
subsequent to receiving information about it, the newly-arriving
myopic receiver chooses an action from a finite set, generates
payoffs, and then leaves the system forever. The sender seeks to
persuade this stream of receivers into choosing actions that are
aligned with her preference by selectively sharing information about
the state at each round.

To tackle the practical challenge of making recommendations in the
absence of distributional data, we depart from the standard Bayesian
persuasion setting and consider situations where neither the sender
nor the receiver knows the distribution of the payoff relevant state.
Instead, the sender learns this distribution over time by observing
the state realizations. We adopt the assumption common in the
literature on Bayesian persuasion that the sender commits to a
signaling mechanism that, at each time step, maps the realized state
to a possibly random {\em action recommendation}. Such a commitment
assumption is well-justified for settings of interest to this work
since online platforms typically design and implement the information
sharing policy as software in advance, rendering frequent changes
unlikely. This advance design serves as a commitment device
organically.

Certainly, the sender cannot freely make arbitrary recommendations, if
the expectation is that these recommendations would influence the
receivers' actions. A natural requirement is for the sender to make
recommendations that the receiver will find optimal to follow, i.e.,
recommendations that are {\em persuasive}. This incentive
compatibility requirement can be easily justified by an application of
the revelation principle. In the case where the sender and the
receivers know the state distribution, the persuasiveness requirement
implies that, subsequent to each recommendation, the recommended
action maximizes the receiver's expected utility under the conditional
state distribution (given the recommendation). However, in the absence
of such distributional knowledge, it is not immediately clear how to
impose persuasiveness.

Our main modeling contribution addresses this issue by proposing a
natural criteria for persuasiveness when neither the sender nor the
receivers know the state distribution. The starting point of our
approach is the observation that any persuasiveness criteria directly
corresponds to a model of receivers' response on receiving a
recommendation (just as in the case of known state distribution).
Thus, by considering reasonable behavioral models for the receiver, we
develop in Section~\ref{sec:persuasiveness} a persuasiveness criterion
that centers {\em robustness\/} as a requirement in the face of
uncertainty. Specifically, our criterion requires that the sender's
recommendations are persuasive under all state distributions in a set
of ``confidence regions'' which contain the true distribution with a
given degree of confidence; these confidence regions shrink over time
as the sender observes more state realizations. This is in line with
the approach in statistics that uses confidence regions to address the
uncertainty in parameter estimates. Furthermore, this robustness
requirement naturally leads to conservative recommendations, thereby
making it likely that the recommendations will be accepted. We refer
to this notion as $\beta$-robustly persuasiveness where $1-\beta$
denotes the confidence level.

\subsection{Algorithmic contribution and regret characterization}

A sender who simply recommends the receiver's best action at the
realized state will certainly be persuasive with complete confidence
($\beta=0$), but may end up with a significant loss in her utility
when compared to her utility had she known the state distribution.
However, since the sender observes the state realizations over time,
she has the opportunity to make more profitable recommendations with
greater confidence in their persuasiveness as she obtains more
information. Thus, the sender's goal is to carefully manage this
tradeoff between the confidence in persuasiveness and her utility, and
achieve low regret against the optimal signaling mechanism with the
knowledge of the state distribution.

The primary theoretical contribution of this work is an efficient
algorithm that, with high probability, makes persuasive
recommendations and at the same time achieves vanishing average
regret. The algorithm we propose proceeds by maintaining at each time
a set of candidate state distributions, based on the observed state
realizations in the past. The algorithm then chooses a signaling
mechanism that is simultaneously persuasive for each of the candidate
distributions and maximizes the sender's utility. Due to this aspect
of the algorithm, we name it the \emph{Robustness against Ignorance}
(\leap) algorithm.

By a careful choice of the candidate set of distributions at each time
period, we show in Theorem~\ref{thm:persuasiveness} that the \leap
algorithm satisfies the $\beta$-robustly persuasiveness criterion for
$\beta = o(T)$, where $T$ is the horizon length. Furthermore,
exploiting the structure of the problem, we show in
Proposition~\ref{prop:computational-efficiency} that the \leap algorithm involves
solving a polynomially-sized (in number of states and actions) linear
program at each period. Taken together, these results establish our
algorithm's persuasiveness and its computationally efficiency.

To characterize the regret of the \leap~algorithm, we next
  undertake a brief digression, in Section~\ref{sec:cost-of-rp}, into
  studying the (static) problem of robust persuasion. Specifically, we
  study a static persuasion setting with known state distribution, but
  impose the restriction that the signaling mechanism must be
  persuasive for all distributions in the neighborhood of the actual
  state distribution. For this problem, we define and analyze a
  quantity $\gap$ that measures the sender's {\em cost of robust
    persuasion}. Formally, $\gap(\mu, \mathcal{B})$ captures the loss
  in the sender's expected utility (under distribution $\mu$) from
  using a signaling mechanism that is persuasive for all distributions
  in the set $\mathcal{B}$, as opposed to using one that is persuasive
  only for the distribution $\mu$. In
  Proposition~\ref{prop:upper-bound}, we establish that, under some
  regularity conditions, the sender's cost of robust persuasion
  $\gap(\mu, \mathcal{B})$ is at most linear in the radius of the set
  $\mathcal{B}$. This is achieved via an explicit construction of a
  signaling mechanism that is persuasive for all distributions in
  $\mathcal{B}$ and achieves sender's utility close to the optimum.
  Further, we provide a matching lower bound in
  Proposition~\ref{prop:lower-bound} by carefully crafting a
  persuasion instance and using its geometry to prove a linear cost of
  robust persuasion; this instance thus serves as a lower bound
  example for robust persuasion. The characterization of the
  \emph{cost of robust persuasion} provides useful insight about the
  problem of robust persuasion, which may be of independent interest.

Using this characterization of the cost of robust persuasion, we
perform a tight regret analysis of persuasion under unknown state
distribution in Section~\ref{sec:regret-analysis}. Our positive
result, Theorem~\ref{thm:vanishing-regret}, establishes that for any
persuasion setting satisfying the aforementioned regularity
conditions, the \leap algorithm achieves $O(\sqrt{T\log T})$ regret
with high probability. Furthermore, in
Theorem~\ref{thm:regret-lower-bound}, we provide a matching
lower-bound (up to $\log T$ terms) for the regret of any algorithm
that makes persuasive recommendations. In addition to the
characterization of $\gap$ and the custom persuasion instance from
Propositions~\ref{prop:upper-bound} and~\ref{prop:lower-bound}, the
proofs of these theorems rely on concentration results for sums of
independent random vectors in Banach spaces.

Our results contribute to the work on online learning that seeks to
evaluate the value of knowing the underlying distributional parameters
in settings with repeated interactions~\citep{kleinbergL03}. In
particular, our results fully characterize the sender's \emph{value of
  knowing the state distribution} for repeated persuasion. Our
well-motivated approach to relax the strong assumption of complete
distributional knowledge in the standard persuasion setting is also
aligned with the prior-independent mechanism design
literature~\citep{dhangwatnotai2015revenue,chawla2013prior}.

%
%
%
%
 %
\subsection{Literature Survey}
\label{sec:literature-survey}

Our paper contributes to the burgeoning literature on Bayesian
persuasion and information design in economics, operations research
and computer science. We refer readers to
\citep{kamenica2011bayesian,bergemann2019information} as well as \citep{candogan2020information} for
a general overview of the recent developments and \citep{dughmi2017algorithmicsurvey} for a survey from algorithmic perspective. 

\textbf{Online learning \& mechanism design.} Our
work subscribes to the recent line of work that studies the interplay of
learning and mechanism design in incomplete-information settings, in
the absence of common knowledge on the prior. We briefly discuss
the ones closely related to our work.

\citet{castiglioni2020online} focus on persuasion setting with a
commonly known prior distribution of the state but unknown receiver
types chosen \emph{adversarially} from a finite set. They show that
effective learning, in this case, is computationally intractable but
does admit $O(\sqrt{T})$ regret learning algorithm, after relaxing the
computability constraint. Our model complements theirs by focusing on
known receiver types but unknown state distributions in a
\emph{stochastic} setup. Moreover, we achieve a similar (and tight)
regret bound through a computationally \emph{efficient}
algorithm. Also relevant to us is the recent line of work on Bayesian
exploration
\citep{kremer2014implementing,mansour2015bayesian,mansour2016bayesian}
which is also motivated by online recommendation systems.  In contrast
to our setting, these models assume the prior is commonly known but
the realized state is unobservable and thus needs to be learned during
the repeated interactions.

Dispensing with the common prior itself, \citet{camara2020mechanisms}
study an adversarial online learning model where both a mechanism
designer and the agent learn about the states over time. The agent is
long-lived and is assumed to minimize her counterfactual (internal)
regret in response to the mechanism designer's policy, which is
assumed to be non-responsive to the agent's actions. The authors use a
reinforcement learning approach to mechanism design and characterize
the policy regret of the mechanism designer, taking into account the
agents' responses, relative to the best-in-hindsight fixed
mechanism. Similar to our work, the regret bounds require the
characterization of a ``cost of robustness'' of the underlying design
problem. While related, the receivers in our model are short-lived and
myopic. Furthermore, our model is stochastic rather than adversarial,
and thus a prior exists in our model. More broadly, our model is
similar in spirit to the prior-independent mechanism design
literature~\citep{dhangwatnotai2015revenue,chawla2013prior}, though
our setup is different. Moreover, our algorithm is measured by the
regret whereas approximation ratios are often adopted for
prior-independent mechanism design.

Recent works by \citet{hahn2019secretary,hahn2020prophet} study
information design in online optimization problems such as the
secretary problem~\citep{hahn2019secretary} and the prophet
inequalities~\citep{hahn2020prophet}, and propose
constant-approximation persuasive schemes. These online optimization
problems often take the adversarial approach, which is different from
our stochastic setup and learning-focused tasks. Therefore, our
results are not comparable.

\textbf{Robust persuasion:} The algorithm we propose relies crucially on robust
persuasion due to the ignorance of the prior, and as a part of
establishing the regret bounds for the algorithm, we quantify the
sender's cost of robustness. %
\citet{kosterina2018persuasion} studies a persuasion setting in the
absence of the common prior assumption. In particular, the sender has
a known prior, whereas only the set in which the receiver's prior lies
is known to the sender. Furthermore, the sender evaluates the expected
utility under each signaling mechanism with respect to the worst-case
prior of the receiver. Similarly, \citet{hu2020robust} study the
problem of sender persuading a privately informed receiver, where the
sender seeks to maximize her expected payoff under the worst-case
information of the receiver. Finally, \citet{dworczak2020preparing}
study a related setting and propose a lexicographic solution concept
where the sender first identifies the signaling mechanisms that
maximize her worst-case payoff, and then among them chooses the one
that maximizes the expected utility under her conjectured prior. In
contrast to these work, our model focuses on a setting with common,
but unknown, prior, and where the receiver has no private
information. Instead, our notion of robustness is with respect to this
unknown (common) prior.

\textbf{Safe online learning:} Our work also relates to safe
  online learning. The work by \citet{moradipari2021safe} is the most
  relevant to our work. They study a safe online learning problem
  where the linear reward and a single linear constraint depend on
  different unknown parameters. The learner has access to both the
  reward and the side information about the safety set. In this
  setting, they propose an algorithm based on linear Thompson Sampling
  and achieve the regret $O(\sqrt{T \log^3 T})$. The key difference is
  that their analysis relies on the assumption that a known safe
  action is an interior point of the safety set for all possible
  values of the unknown parameter. Under our regularity conditions, it
  is true that for every distribution there exists a signaling
  mechanism for which all the persuasiveness constraints hold strictly
  (that is, the order of the quantifiers from above is interchanged).
  However, it is unclear if this weaker assumption would be sufficient
  for their setting.

  \citet{amani2019linear} study a linear stochastic multi-armed bandit
  problem where the linear reward function and a \emph{single} linear
  safety constraint depend on an unknown parameter. Their main
  algorithm and its analysis depend on knowing (a lower bound on) the
  safety gap, i.e., the slack in the safety constraint for the optimal
  solution under the true parameter. When the safety gap is known and
  positive (i.e., the constraint is inactive), they prove a regret of
  $O(\log T \sqrt{T})$. On the other hand, if the safety gap is known
  to be zero, they only achieve a regret of $\Tilde{O}(T^{2/3})$. They
  provide a separate algorithm for the case of an unknown safety gap
  and state a regret bound of $\Tilde{O}(T^{2/3})$. In our setting,
  there are \emph{multiple} persuasiveness constraints, and many of
  these would be active for the true distribution in nontrivial
  settings. Thus even if their work can be extended to multiple
  constraints, it may only guarantee $\Tilde{O}(T^{2/3})$ regret
  bound.

  \citet{usmanova2019safe} seek to minimize a smooth convex function
  over a set of uncertain linear constraints where both the
  coefficients and constant parameters are unknown. Although our
  problem is a specific case of theirs, our model does not meet their
  central assumption of being able to evaluate the constraints at any
  point within a small neighborhood of the feasible set. 

  Recent works by \citep{pacchiano2021stochastic, Khezeli2020SafeLS,
    moradipari2020stage, moradipari2021safe} study a similar safe
  learning problem in different contexts.
  \citet{pacchiano2021stochastic} require that at each time, the
  chosen action has an expected cost below a certain threshold.
  \citet{Khezeli2020SafeLS, moradipari2020stage} study safe learning
  where in addition to maximizing the expected reward, one requires
  the reward to be above a threshold with high probability. In these
  settings, the objective and the constraint are aligned. Our setup is
  different because the sender's and the receivers' preferences,
  corresponding respectively to the objective and constraints, need
  not be aligned with each other. Most importantly, all these work
  impose a single constraint at each round, whereas our persuasiveness
  condition requires multiple constraints at each round.

  \textbf{Online linear/convex optimization:} Since the persuasion
  problem can be posed as a linear program, our work also relates to
  the online convex optimization problem. Mostly, the focus here is on
  adversarial setting where the loss function (objective) is
  adversarially chosen and revealed at the end of each time period.
  Some papers~\citep{cao2019time,mahdavi2013stochastic} focus on the
  stochastic setting, but either study an unconstrained
  problem~\citep{cao2019time} or study a batch algorithm rather than
  an online algorithm~\citep{mahdavi2013stochastic}. Focusing on the
  constraints, and using the terminology of~\citep{kim2023online},
  these work typically consider either a {\em long-term constraint}
  formulation~\citep{yu2017online, mahdavi11, neely2017online,
    yi2021regret, kim2023online, cao2018online}, or consider a {\em
    cumulative constraint} formulation~\citep{yuan18, yi2022regret,
    guo2022online}. The \emph{long-term constraint} formulation
  requires feasibility on average in the long run. Such constraints
  are reasonable in applications where the constraints are on
  aggregate quantities, such as budgets in online
  advertising~\citep{liakopoulos2019cautious}, covering constraints in
  sensor networks, capacity constraints in online
  routing~\citep{agrawal2014bandits}, etc. However, this type of
  constraint is not reasonable in our setting as it would permit the
  sender to make poor recommendations in some rounds as long as it can
  be compensated by good recommendations in other rounds. In contrast,
  the \emph{cumulative} constraint formulation focuses on bounding the
  sum of the positive-parts of the constraints (which require some
  quantity to be non-positive). This formulation is equivalent to our
  formulation if the cumulative constraint can be made zero. However,
  most previous work allow for some constraint violation and seek to
  bound the order of the violations. In the presence of such
  violations, our formulation is stronger.

     Finally, by characterizing the
persuasion problem as a Stackelberg game between the sender's choice
of a signaling mechanism and the receiver's subsequent choice of an
action, our work is related to the broader work on the
characterization of regret in repeated Stackelberg
settings~\citep{balcan2015commitment, dong2018strategic,
  chen2019learning}.
%
%
%
%
%
%

%
%
%
%
%
%
%

%
%
%
%
%
%

%
%
%
%
%
%
%
%
%

%
%
%
%
 %
\section{Model}
\label{sec:model}

Consider a persuasion setting with a single long-run \textit{sender}
persuading a stream of homogeneous \textit{receivers} who
arrive sequentially over a time horizon of length $T$. At each time
$t \in [T] = \{0, \cdots, T-1\}$, a state $\omega_t \in \Omega$ is
drawn independently and identically from a state distribution
$\mu^* \in \Delta(\Omega)$. (Here, for any finite set $X$, $\Delta(X)$
denotes the set of all probability distributions over $X$.) We focus
on the setting where $\Omega$ is a known finite set, however the
distribution $\mu^*$ is \textit{unknown} to both the sender and the
receivers. To capture the sender's initial knowledge (before time
$t=0$) about the distribution $\mu^*$, we assume that the sender knows
that $\mu^*$ lies in the set $\mathcal{B}_0 \subseteq \Delta(\Omega)$.

At each time $t \in [T]$, the sender observes the realized state
$\omega_t$, and shares with the arriving receiver an \textit{action
  recommendation} $a_t \in A$ (chosen according to a \textit{signaling
  algorithm}, as described below), where $A$ is a finite set of
actions available to the receivers. The receiver then chooses an
action $\hat{a}_t \in A$ (not necessarily equal to $a_t$). This
results in the receiver obtaining a utility $u(\omega_t, \hat{a}_t)$
and the sender obtaining a utility $v(\omega_t, \hat{a}_t)$. Without
loss of generality, we assume that $v(\omega, a) \in [0,1]$ for all
$\omega \in \Omega$ and $a \in A$. Further, to avoid trivialities, we
assume $|\Omega| \geq 2$ and $|A|\geq 2$. We refer to the tuple
$\instance = (\Omega, A, u, v, \mathcal{B}_0)$ with
$u: \Omega \times A \to \reals$ and $v: \Omega \times A \to [0,1]$ as
an \textit{instance} of our problem.

Before we proceed, we make few remarks on the persuasion instance.
First, the preceding description does not specify a model of the
receivers' actions $\hat{a}_t$. As we discuss below in
Section~\ref{sec:persuasiveness}, this issue is intertwined with the
{\em persuasiveness constraints\/} that we impose on the sender's
signaling algorithm, and hence, we postpone the discussion until then.
Second and relatedly, while we have assumed that that sender shares
information in the form of actions recommendations, under the
persuasiveness constraints we consider it can be shown that this is
without loss of generality. Third, while our definition of an instance
assumes that the receivers are homogeneous, it can be extended to
allow for heterogeneity of receivers' utility; our results continue to
apply in the setting where the receivers' types are observable to the
sender. Finally, we assume that the sender knows the
  receivers' utility. This is justified in the context of our
  applications of interest, namely online platforms, where given the
  scale, the platform may have good estimates about user utility from
  extensive user-level data.

Informally, given a persuasion instance $\instance$, the sender's goal
is to systematically make action recommendations such that her
long-run total utility is maximized. We now describe the formal
algorithmic aspects of the sender's goal.

As each time $t$, the sender chooses an action recommendation $a_t$
based on the past state realizations, the past action recommendations
as well as the past actions chosen by the receivers. To separate the
historical information from that about the present, we define the {\em
  history} $h_t$ at the beginning of time $t$ as follows:
$h_t = \cup_{\tau < t} \{ (\tau, \omega_\tau, a_\tau, \hat{a}_\tau)\}$
(with $h_0 = \emptyset$), and note that the sender observes
$(h_t, \omega_t)$ prior to making the recommendation $a_t$ at time
$t$. We also note that, since the receivers do not know the state
distribution $\mu^*$, neither the past actions recommended by the
sender nor the past actions chosen by the receivers carry any
information about $\mu^*$ beyond that contained in the state
realizations. Thus, the part of the history that is relevant to the
sender consists of only the state realizations until time $t$.

A \textit{signaling algorithm} $\alg \equiv \alg(\instance)$ for the
sender specifies, at each time $t \in [T]$ and after any history $h_t$
and state $\omega_t$, a probability distribution
$\sigma^\alg(h_t, \omega_t, \cdot) \in \Delta(A)$ over the set of
actions. (We sometimes drop the superscript $\alg$ when it is clear
from the context.) Specifically, once the state $\omega_t$ is
realized, the sender draws the action recommendation $a_t$
independently according to the distribution
$\sigma(h_t, \omega_t, \cdot) \in \Delta(A)$. Thus, the probability
that the sender recommends an action $a \in A$ is given by
$\sigma(h_t, \omega_t, a)$. Implicitly, the notion of a signaling
algorithm reflects the assumption that the sender {\em commits\/} to a
mechanism for sending recommendations.

Given an instance $\instance$ and a signaling algorithm
  $\alg$, the sender's total (realized) utility is given by
  \begin{align*}
    V_\instance(\alg, T) \defeq \sum_{t \in [T]} v(\omega_t, \hat{a}_t).
  \end{align*}
  Thus, to evaluate the performance of a signaling algorithm, we need
  a model of the receivers' response subsequent to receiving the
  action recommendations. Rather than directly specifying such a
  response model, we instead model conditions on the signaling
  algorithm $\alg$ which result in {\em obedient\/} responses from the
  receivers, i.e., which lead each receiver to choose the action
  recommended: $\hat{a}_t = a_t$.
  Any such condition on the signaling algorithm $\alg$ implies a model
  of receivers' response, and the converse can be assumed without loss
  of generality by invoking incentive compatibility and the revelation
  principle. Henceforth, we refer to such a condition as a {\em
    persuasiveness criterion\/}.

  To motivate these persuasiveness criteria on the signaling
  algorithms, we first discuss the setting where the sender and the
  receivers commonly know the state distributions. This setting will
  also serve as a benchmark to compare the performance of any
  signaling algorithm satisfying certain persuasiveness requirements.

\subsection{Benchmark: Known State Distribution}%
\label{sec:known-dist-benchmark}

Consider the setting where the sender and the receivers commonly know
the state distribution $\mu^* = \mu \in \Delta(\Omega)$. In this
setting, each receiver responds by choosing the action that maximizes
her expected utility under the posterior belief about the state given
the action recommendation. In particular, the sender's problem
decouples across time periods, and standard
results~\citep{kamenica2011bayesian,bergemann2019information,
  dughmi2019algorithmic} imply that the sender's problem at each
period can be formulated as a linear program.

To elaborate, fix a time $t \in [T]$ and history $h_t$, and consider
the persuasion problem between the sender and the arriving receiver.
Recall that $\sigma(h_t, \omega, a)$ denotes the probability with
which the sender recommends action $a_t = a$ if the realized state is
$\omega_t = \omega$. We refer to
$\sigma[h_t] \defeq ( \sigma(h_t, \omega, a) : \omega \in \Omega, a\in
A)$ as the {\em signaling mechanism\/} at time $t$, and drop the
dependence on $h_t$ if the context is clear. Finally, let
$\mathcal{S} = \{ \sigma : \sigma(\omega, \cdot) \in \Delta(A) \text{
  for each $\omega \in \Omega$}\}$ denote the set of all signaling
mechanisms.

A signaling mechanism $\sigma \in \mathcal{S}$ is \textit{persuasive},
if conditioned on receiving an action recommendation $a \in A$, it is
indeed optimal for the receiver to choose action $a$. Let $a \in A$ be
an action with
$\sum_{\omega \in \Omega} \mu(\omega)\sigma(\omega, a) > 0$. Upon
receiving the recommendation $a$, the receiver's posterior belief that
the realized state is $\omega$ is given by Bayes' rule as
$\tfrac{\mu(\omega) \sigma(\omega, a)}{\sum_{\omega' \in \Omega}
  \mu(\omega') \sigma(\omega', a)}$, and hence
$\sum_{\omega\in \Omega} \left( \frac{\mu(\omega) \sigma(\omega,
    a)}{\sum_{\omega' \in \Omega} \mu(\omega') \sigma(\omega',
    a)}\right) u(\omega, a')$ denotes her expected utility of choosing
action $a' \in A$ conditioned on receiving the recommendation $a$. For
the receiver's expected utility to be maximized from choosing action
$a$, we need
$\sum_{\omega \in \Omega} \mu(\omega) \sigma(\omega, a) \left(
  u(\omega, a) - u(\omega, a')\right) \geq 0$ for all $a' \in A$.
Since the inequality is trivially satisfied if
$\sum_{\omega \in \Omega}\mu(\omega)\sigma(\omega,a) = 0$, the set of
persuasive mechanisms $\pers(\mu)$ is given by
\begin{align}
  \label{eq:persuasive-signals}
  \pers(\mu) \defeq \left\{ \sigma \in \mathcal{S} : \sum_{\omega \in
  \Omega} \mu(\omega) \sigma(\omega, a) \left( u(\omega, a) -
  u(\omega, a')\right) \geq 0, ~\text{for all $a, a' \in A$}\right\}.
\end{align}
We note that the set $\pers(\mu)$ is a convex polytope for all
$\mu \in \Delta(\Omega)$. Furthermore, the set $\pers(\mu)$ is
non-empty, since it always contains the ``full-information mechanism''
which recommends the receiver's optimal action at each state.

Given a persuasive signaling mechanism $\sigma \in \pers(\mu)$, the
receiver is incentivized to choose the recommended action. Assuming
ties are broken in favor of the recommended action, the sender's
expected utility is given by
\begin{align*}
  V(\mu, \sigma) &\defeq \sum_{\omega \in \Omega} \sum_{a \in A} \mu(\omega) \sigma(\omega, a) v(\omega, a).
\end{align*}
Since $V(\mu, \sigma)$ is linear in $\sigma$, the problem of selecting
an optimal persuasive signaling mechanism is given by the following
linear program:
\begin{align}\label{opt:known-prior-lp}
  \opt_\instance(\mu) \defeq  \max_\sigma V(\mu, \sigma),\ \text{subject to $\sigma \in \pers(\mu)$}.
\end{align}
Finally, letting $\sigma^*$ denote an optimal signaling mechanism to
the preceding optimization problem, the algorithm $\alg$ that sets
$\sigma^\alg(h_t, \omega_t,a) = \sigma^*(\omega_t, a)$ after any
history $h_t$ optimizes the sender's total expected utility when the
state distribution is known, with total expected utility given by
$T\cdot \opt_\instance(\mu)$.

\subsection{Persuasiveness Criterion: Unknown Distribution}%
\label{sec:persuasiveness}

We now return to the setting with unknown state distribution, and
discuss refined persuasiveness conditions on the signaling algorithm
under which the receivers' response can be reasonably assumed to equal
the recommendation. In particular, we propose and motivate a condition
on the signaling algorithm, namely the {\em robust persuasiveness\/}
criterion as described in Definition~\ref{def:robust-persuasiveness},
and provide detailed justification supporting the notion.

We begin with the simplest criterion inspired from the known
distribution setting. As the sender observes the past state
realizations, the empirical distribution $\gamma_t$, with
$\gamma_t(\omega) \defeq \frac{1}{t} \sum_{\tau < t }\ind\{
\omega_\tau = \omega\}$, provides an estimate for the unknown
distribution $\mu^*$. A natural first idea, which we call the {\em
  naive criterion\/}, simply requires the algorithm to act as if this
estimate is exact:
\begin{definition} A signaling algorithm $\alg$ satisfies the
  \emph{naive criterion} if each $\sigma^\alg[h_t]$ is persuasive
  under the empirical distribution at time $t$, i.e.,
  $\sigma^\alg[h_{t}] \in \pers(\gamma_t)$ for all $ t \in [T]$.
\end{definition}

The naive criterion can be motivated through a particular behavioral
model of the receivers involving social learning. Specifically,
consider a platform setting where each receiver (i.e., a user) arrives
with an uninformative Haldane prior~\citep{haldane,villegas,jaynes}
over the state distribution $\mu^*$, and observes all the past state
realizations. The latter holds if we assume there is social learning
among the receivers, where each receiver leaves a feedback that is
read by all subsequent receivers. Then, at each time $t$ the
corresponding receiver's belief about the state would be exactly the
empirical distribution $\gamma_t$, and thus the receiver would
optimally accept the recommendation made by the platform if it uses a
signaling algorithm satisfying the naive criterion.

However, from a practical perspective, the preceding model makes very
restrictive assumptions. First, in a platform setting, the users'
prior belief over $\mu^*$, if such a prior exists at all, is unlikely
to be known to the platform, and need not be same across different
users (let alone be the uninformative Haldane prior). Second, even
with social learning, the users typically would not observe all the
past state realizations (or even just the empirical distribution);
this is because not all users leave reviews in a platform, and a user
would typically read only a subset of available reviews. Thus, under a
realistic model of social learning, the receivers' belief about the
state would be in general different from the empirical distribution. 

In addition to relying on restrictive behavioral assumptions, there
are other deficiencies with the naive criterion that render it
ill-suited as a criterion for ensuring persuasiveness. First, the
naive criterion is especially weak in the initial stages of persuasion
due to the lack of sufficient data; at these initial stages, the
constraint based on the empirical distribution may not constrain the
sender's recommendations. For instance, if the empirical distribution
at the beginning happens to be skewed and concentrates on very few
states, then the naive criterion imposes no restriction on the action
recommendations at any previously unseen state since it has zero
empirical probability. Second, an algorithm satisfying the naive
criterion may still make {\em inconsistent} recommendations across
time. That is, for such an algorithm, there may not exist a single
belief $\mu$ for which the recommendations as a whole are persuasive,
i.e., $\sigma^\alg[h_t] \in \pers(\mu)$ for all $t$. Any such belief
$\mu$, if it exists, provides a justification for the signaling
algorithm, and larger the set of such beliefs the stronger is the
justification. For instance, the ``full-information'' signaling
algorithm $\full$, which always recommends the receivers' best action
$a_t \in \argmax_{a \in A} u(\omega_t, a)$ after any history $h_t$,
has the strongest justification since all beliefs
$\mu \in \Delta(\Omega)$ satisfy $\sigma^\full[h_t] \in \pers(\mu)$.
On the other hand, one can easily construct examples where an
algorithm satisfying naive criterion fails to have even a single
belief justifying it, due to inconsistencies in recommendations across
different periods.

Summarizing, the primary reason for the weaknesses of the naive
criterion is its reliance on the point estimate $\gamma_t$ in the
place of receivers' inherently uncertain beliefs about the state. Even
for basic inferential tasks, such point estimates are seldom
sufficient. Without explicitly incorporating this uncertainty into its
conditions, an algorithm would provide no {\em confidence} that the
receivers will accept and act according to the recommendations. To
remedy these weaknesses, we propose the following criterion that
embraces the notion of robustness in its conditions.
\begin{definition}\label{def:robust-persuasiveness} Given $\beta \geq 0$, a signaling algorithm $\alg$
  is \emph{$\beta$-robustly persuasive}, if there exists
  (history-dependent) sets $\mathcal{C}_t \subseteq \mathcal{B}_0$ for
  all time $t$, such that
  \begin{enumerate}
  \item \textbf{Robustness:} The signaling mechanism
    $\sigma^\alg[h_t]$ is persuasive for \emph{all} beliefs in the set
    $\mathcal{C}_t$: for each $t \in [T]$, we have
    \begin{align*}
      \sigma^\alg[h_t] \in \pers(\mathcal{C}_t) \defeq \cap_{\mu \in \mathcal{C}_t} \pers(\mu).
    \end{align*}
      \item \textbf{Coverage:} The sets $\mathcal{C}_t$ all contain the true state
    distribution $\mu^*$ with high probability:
    \begin{align*}
      \prob_{\mu^*} \left( \cap_{t \in [T]} \mathcal{C}_t \ni \mu^*\right) \geq 1 -\beta.
    \end{align*}
    (Here, $\prob_{\mu^*}$ represents the probability with respect to
    the (unknown) distribution $\mu^*$ and any independent
    randomization in the algorithm.)
  \end{enumerate}
\end{definition}
The first condition in the criterion enforces robustness, requiring
that the signaling mechanism at time $t$, $\sigma^\alg[h_t]$, is
persuasive with respect to {\em all} beliefs in the set
$\mathcal{C}_t$. These sets implicitly capture the uncertainties
regarding the receivers' beliefs, and by depending on the history,
reflect any learning occurring over time. (We note that the set
$\pers(\mathcal{C}_t)$ is indeed non-empty, as it contains the
``full-information'' mechanism.) The second condition in the criterion
requires these sets to have good coverage properties, i.e., these sets
contain the state distribution $\mu^*$ with high probability.

To further motivate the criterion, we delve a bit into the perspective
of social learning in a platform setting mentioned earlier. Here,
while it is a strong assumption to require the receivers to know the
exact empirical distribution, it is fair to assume that the receivers
observe (summary statistics about) a sizeable proportion of past state
realization. In particular, many common empirical principles, such as
the ``90-9-1 rule''~\citep{antelmi2019characterizing,van20141}, posit
that a constant fraction of the users leave feedback in the platform.
In this context, a receiver who starts with some sufficiently diffuse
prior over $\mu^*$, and who learns from past (incomplete) feedback,
will have a belief about the state that is close enough to the
empirical distribution. Thus, a signaling algorithm that makes
recommendations that are persuasive for {\em all} beliefs close to the
empirical distribution would ensure that such a receiver would find it
optimal to follow the recommended action. Our proposed criterion, by
using a robustness approach, abstracts away from the details of such
an explicit model, and captures the receivers' response through the
uncertainty sets $\mathcal{C}_t$.

Observe that as long as the sets $\mathcal{C}_t$ contain the empirical
distribution $\gamma_t$, the preceding criterion is stronger than the
naive criterion. More importantly, in addition to capturing more
realistic models of social learning, the coverage and the robustness
conditions together also overcome the other inadequacies of the naive
criterion that we discussed above. To see this, note that, at the
initial stages $t$ when the data is insufficient, good coverage
requires the set $\mathcal{C}_t$ to be large, and thus the action
recommendations are severely constrained (even at the states that have
not been realized), unlike the case with the naive criterion.
Similarly, the robustness ensures that any belief
$\mu \in \cap_{t \in [T]} \mathcal{C}_t$ provides a justification for
the signaling algorithm, thus precluding any inconsistencies across
time. In particular, with probability at least $1- \beta$, the true
state distribution $\mu^*$ justifies all the recommendations made by a
$\beta$-robustly persuasive signaling algorithm:
$\prob_{\mu^*}\left(\sigma^\alg[h_t] \in \pers(\mu^*) \text{ for all
    $t \in [T]$}\right) \geq 1- \beta$.

The parameter $\beta$ in the criterion plays the same role as that
played by significance level in inference. In particular, low values
of $\beta$ correspond to high level of confidence in the uncertainty
sets $\mathcal{C}_t$. Finally, it is easy to see that $\beta$-robustly
persuasive algorithms exist for any $\beta \geq 0$; in fact, choosing
the sets $\mathcal{C}_t = \mathcal{B}_0$ for all $t \in [T]$, it
follows that the algorithm $\full$ is $0$-robustly persuasive.

Given the preceding discussion, we hereafter assume that for any
signaling algorithm $\alg$ that is $\beta$-robustly persuasive for
some (small) $\beta \geq 0$, the receivers' response $\hat{a}_t$
equals the action recommendation $a_t$ at each time $t$. Thus, for any
such algorithm $\alg$, the sender's total utility reduces to
$V_{\instance}(\alg, T) = \sum_{t \in [T]} v(\omega_t, a_t)$.

\subsection{Sender's Learning Problem}

Finally, we describe the evaluation metric for the performance of any algorithm satisfying the
preceding persuasiveness criterion by comparing the sender's utility
$V_\instance(\alg, T)$ against the known-distribution benchmark given
by $T \cdot \opt(\mu^*)$. Specifically, we measure the sender's
\textit{regret} from using a $\beta$-robustly persuasive algorithm
$\alg$ by
\begin{align}
  \label{eq:regret}
  \regret_\instance(\alg, T, \mu^*) &\defeq T \cdot \opt_\instance(\mu^*)   - V_\instance(\alg,  T) = T \cdot \opt_\instance(\mu^*)   - \sum_{t \in [T]} v(\omega_t, a_t).
\end{align}
We are now ready to formalize the sender's learning problem. Begin by
noticing that one must require the signaling algorithm $\alg$ to be
$\beta$-robustly persuasive for some small $\beta$ in order for the
second equality above to hold, i.e., for the receivers' responses to
match the recommendations. At the same time, $0$-robustly
persuasiveness is an excessive requirement, with no hope of resulting
in a sub-linear regret. (In Appendix~\ref{ap:0-persuasive}, we present
an example instance where any $0$-robustly persuasive algorithm
necessarily obtains a linear regret.) Thus, the central problem is to
design, for any given instance $\instance$, an algorithm $\alg$ that
is $\beta$-robustly persuasive for small (vanishing) $\beta$ and
simultaneously achieves sublinear regret with high probability.

%
%
%
%
 %
%
\section{The Robustness Against Ignorance (\leap) Algorithm}%
\label{sec:leap}

Having described the learning problem faced by the sender, in this
section, we present a signaling algorithm that we call the Robustness
Against Ignorance (\leap) algorithm. Here, we show that the
\leap~algorithm is $\beta$-robustly persuasive with $\beta = o(1)$,
relegating the regret analysis to Section~\ref{sec:regret-analysis}.

Before describing our proposed algorithm, we briefly motivate our
design approach. Observe that if the state distribution $\mu^*$ is
known, then the sender's problem is given by the linear
program~\eqref{opt:known-prior-lp}, and thus the optimal signaling
mechanism can be efficiently computed. Thus, a natural learning
approach is to solve at each time $t$ the estimated version of the
LP~\eqref{opt:known-prior-lp}, where the unknown distribution $\mu^*$
is replaced by the empirical distribution $\gamma_t$, and use the
corresponding optimal signaling mechanism for that time period.
However, this alone is not sufficient to obtain an algorithm that is
$\beta$-robustly persuasive, which requires the signaling mechanisms
to be persuasive for all distributions in some small neighborhood of
$\mu^*$. To elaborate, simply solving the estimated LP may yield
solutions that are only $\epsilon$-feasible for distributions close to
the empirical distribution, i.e., some of the persuasiveness
constraints for such nearby distributions may get violated. In fact,
optimizing the estimated LP may result in a mechanism that is not
persuasive for {\em any} other distribution close to the empirical
distribution. Thus, an immediate challenge is in determining how to
use the empirical distribution estimate to find well-performing
signaling mechanisms that are persuasive (with high probability) for
{\em all} distributions in a small neighborhood around the unknown
state distribution. Part of this challenge is to carefully choose the
corresponding neighborhoods without significantly sacrificing the
performance of the mechanism.

The algorithm we propose is adaptive. An alternative is to adopt an
``explore-then-commit'' design~\citep{lattimoreS2020}, where the
algorithm uses the state realizations in the first $t$ periods (for
some carefully chosen $t$) to estimate the unknown distribution and
subsequently commits to a single signaling mechanism for the remaining
time periods. However, it is unlikely that such a algorithmic design
can achieve strong regret guarantees in our setting, since it is known
that such an approach yields the sub-optimal $O(T^{2/3})$ regret in
simple multi-armed bandit problems~\citep{lattimoreS2020}. This
observation illustrates the need for adaptivity to obtain order-wise
optimal regret.

To meet these challenges, our algorithm $\leap$ proceeds by adaptively
maintaining, at each time $t \geq 0$, a set $\mathcal{B}_{t}$ of
candidates for the (unknown) distribution $\mu^*$. This set is a
(closed) $\ell_1$-ball of radius $\epsilon_t$ at the empirical
distribution $\gamma_t$. It then selects a signaling mechanism that
maximizes the sender expected utility w.r.t.~the empirical estimate
$\gamma_t$ among mechanisms that are persuasive for all distributions
$\mu \in \mathcal{B}_t$. Finally, it makes an action recommendation
$a_t$ using this signaling mechanism, given the state realization
$\omega_t$. The \leap~algorithm is formally described in
Algorithm~\ref{alg:leap-algorithm}. Here, we use the notation
$\pers(\mathcal{B})$ to denote the set of signaling mechanisms that
are simultaneously persuasive under all distributions $\mu$ in the set
$\mathcal{B} \subseteq \Delta(\Omega)$:
$\pers(\mathcal{B}) = \cap_{\mu \in \mathcal{B}} \pers(\mu)$. We
remark that for any non-empty set
$\mathcal{B} \subseteq \Delta(\Omega)$, the set $\pers(\mathcal{B})$
is convex since it is an intersection of convex sets $\pers(\mu)$, and
is non-empty since it contains the full-information signaling
mechanism. Furthermore, we let
$\mathsf{B}_1(\mu, \epsilon) \defeq \{ \mu' \in \Delta(\Omega) : {\|
  \mu' - \mu\|}_1 \leq \epsilon\}$ denote the (closed) $\ell_1$-ball
of radius $\epsilon > 0$ at $\mu \in \Delta(\Omega)$.

\begin{algorithm}[t]
  \SetAlgoNoLine
  \KwIn{Instance $\instance$, Time horizon $T$}
  \KwPara{$\gamma_0 \in \mathcal{B}_0$, $\{ \epsilon_t > 0 : t \in [T]\}$}
  \KwOut{$a_t \in A$ for each $t \in [T]$}
  \For{$t=0$ to $T-1$}{
    Choose any $\sigma[h_{t}] \in \arg\max_\sigma \{ V(\gamma_{t}, \sigma) : \sigma \in  \pers(\mathcal{B}_{t})\}$\;
    Recommend $a_t = a \in A$ with probability
                $\sigma(\omega_t, a; h_t)$\;
    Update $\gamma_{t+1}(\omega) \gets \frac{1}{t+1} \sum_{\tau=0}^{t} \ind\{\omega_\tau = \omega\}$ for each $\omega \in \Omega$\;
    Set $\mathcal{B}_{t+1} \gets \mathsf{B}_1(\gamma_{t+1}, \epsilon_{t+1})$;
  }
  \caption{The Robustness Against Ignorance (\leap) algorithm}
  \label{alg:leap-algorithm}
\end{algorithm}

From the intuitive description, it follows that the sets
$\mathcal{B}_t = \mathsf{B}_1(\gamma_t, \epsilon_t)$ naturally play
the role of the covering sets $\mathcal{C}_t$ in the definition of
$\beta$-robustly persuasiveness. Specifically, the parameters
$\{\epsilon_t : t \in [T]\}$ control the degree of persuasiveness of
the algorithm: larger values of $\epsilon_t$ imply that the algorithm
is $\beta$-robustly persuasive for smaller values of $\beta$. (In
particular, if all $\epsilon_t$ are larger than $2$, the algorithm
reduces to the full-information algorithm $\full$, and is $0$-robustly
persuasive.) Unsurprisingly, larger values of $\epsilon_t$ also lead
to larger regret, and hence the sender must choose $\epsilon_t$ to
optimally trade-off the persuasiveness of the algorithm against its
regret.

Our first main result characterizes \leap's persuasiveness for a
particular choice of parameter values which we show in
Section~\ref{sec:regret-analysis} to be regret-optimal.
\begin{theorem}\label{thm:persuasiveness}  For each $t \in [T]$, let
  $\epsilon_t = \min\{\sqrt{\frac{|\Omega|}{t}} \left(1 + \sqrt{\Phi
      \log T}\right),2\}$ with $\Phi > 0$. Then, the \leap~algorithm
  is $\beta$-robustly persuasive with
  \begin{align*}
    \beta = \sup_{\mu^* \in \mathcal{B}_0} \prob_{\mu^*}\left( \cap_{t
    \in [T]} \mathcal{B}_{t} \not\ni \mu^*\right) \leq T^{1 - \frac{3 \Phi  \sqrt{\Omega}}{56}}.
  \end{align*}
  In particular, for $\Phi > 20$, we have $\beta \leq
  T^{-0.5}$.
\end{theorem}
The proof of the persuasiveness of \leap follows by showing that the
empirical distribution $\gamma_t$ concentrates around the unknown
state distribution $\mu^*$ with high probability. Since, after any history $h_t$,
the signaling mechanism $\sigma[h_t]$ chosen by the algorithm is
persuasive for all distributions in an $\ell_1$-ball around $\gamma_t$, we
deduce that it is persuasive under $\mu^*$ as well. To show the
concentration result, we use a concentration inequality for
independent random vectors in a Banach space~\citep{foucartH13}; the
full proof is provided in Appendix~\ref{ap:persuasiveness-proof}.

We observe that to get strong persuasiveness guarantees, the choice of
$\epsilon_t$ in the preceding theorem requires the knowledge of the
time horizon $T$. However, applying the standard doubling
tricks~\citep{besson2018doubling}, one can convert our algorithm to an
{\em anytime} version that has the same regret upper bound guarantee,
at the cost of a weakened persuasiveness guarantee, where the
persuasiveness $\beta$ is weakened to a constant arbitrarily close to
$0$.

Next, note that the \leap~algorithm requires finding at each time $t$
a signaling mechanism that is persuasive for {\em all} distributions
in a neighborhood around the empirical distribution. The following
result shows that this is a simple computational task requiring a
polynomial running time. Thus, the result establishes the
\leap~algorithm's computational efficiency.
\begin{proposition}\label{prop:computational-efficiency}
  The \leap~algorithm requires solving at each time a linear program
  with size polynomial in $|\Omega|$ and $|A|$.
\end{proposition}
\proof{Proof.} To see the efficiency of the \leap algorithm, note that
at each time $t$ the algorithm has to solve the optimization problem
$\max_\sigma \{ V(\gamma_t, \sigma) : \sigma \in
\pers(\mathcal{B}_t)\}$. Since
$\mathcal{B}_t = \mathsf{B}_1(\gamma_t, \epsilon_t)$ is an
$\ell_1$-ball of radius $\epsilon_t$, it is a convex polyhedron with
at most $|\Omega| \cdot (|\Omega|-1)$ vertices. (These vertices are
all of the form
$\gamma_t + \frac{\epsilon_t}{2}\left(e_\omega - e_{\omega'}\right)$,
where $e_\omega$ is the belief that puts all its weight on $\omega$.)
By the linearity of the obedience constraints and the convexity of
$\mathcal{B}_t$, it follows that $\pers(\mathcal{B}_t)$ is obtained by
imposing the obedience constraints at priors corresponding to each of
these vertices. Since there are $O(|\Omega| + |A|^2)$ obedience
constraints for each distribution, we obtain that the optimization
problem is a polynomially-sized linear program, and hence can be
solved efficiently.~\Halmos\endproof

Having addressed the persuasiveness and the computational efficiency
of the \leap~algorithm, we devote the rest of the paper to analyzing
its regret. To do this, we first take a digression to define (and
bound) the cost of robust persuasion in static persuasion problems.
Armed with this result, we then characterize the algorithm's regret in
Section~\ref{sec:regret-analysis}.

  %
  %
  %
  %
  %
  %
  %

  %
  %
  %

%
%
%
%
 %
\section{Digression: Cost of Robust Persuasion}
\label{sec:cost-of-rp}
In this section, we consider the static persuasion problem with known
state distribution (discussed in
Section~\ref{sec:known-dist-benchmark}), and study the loss in the
sender's expected utility from requiring the signaling mechanism to be
persuasive for all distributions in a neighborhood around the state
distribution. To measure this loss, we first define the notion of the
\textit{cost of robust persuasion}, a quantity that depends on the
neighborhood, and provide upper and lower bounds under some minor
regularity conditions.

Fix a persuasion instance $\instance$. In the static setting with
known state distribution $\mu$, the sender's optimal expected payoff
is given by
$\opt_\instance(\mu) = \sup_{\sigma \in \pers(\mu)} V(\mu, \sigma)$.
Next, for any set of distributions
$\mathcal{B} \subseteq \mathcal{B}_0$, the set of signaling mechanisms
that are simultaneously persuasive for all distributions in
$\mathcal{B}$ is given by
$\pers(\mathcal{B}) = \cap_{\mu' \in \mathcal{B}} \pers(\mu')$. Hence,
the sender's optimal expected utility among all such mechanisms is
given by $\sup_{\sigma \in \pers(\mathcal{B})} V(\mu, \sigma)$. Thus,
we define the {\em cost of robust persuasion} as
\begin{align}\label{eq:gap}
  \gap(\mu, \mathcal{B})
  &\defeq
    \sup_{\sigma \in \pers(\mu)} V(\mu, \sigma) - \sup_{\sigma \in
    \pers( \mathcal{B})} V(\mu, \sigma).
\end{align}  
Thus, $\gap(\mu, \mathcal{B})$ captures the difference in the sender's
expected utility (under $\mu$) between using the optimal persuasive
signaling mechanism for the distribution $\mu$ and using the optimal
signaling mechanism that is persuasive for all distributions
$\mu' \in \mathcal{B}$.

For general persuasion instances, one can show that the cost of robust
persuasion can be severe: in Appendix~\ref{ap:regularity-failure}, we
present a persuasion instance and a distribution $\mu$ such that for
any $\epsilon >0$, the cost of being robustly persuasive for the set
$\mathsf{B}_1(\mu, \epsilon)$ of distributions satisfies
$\gap(\mu, \mathsf{B}_1(\mu, \epsilon)) = \tfrac{1}{2}$. The instance
we present there is pathological, with an action that is optimal for
the receiver at a single unique distribution. To obtain meaningful
insights on the cost of robust persuasion, we seek to exclude such
instances by imposing some regularity condition on the instances.

To state these regularity conditions, we need some notation. For each
action $a \in A$, let $\mathcal{P}_a$ denote the set of state
distributions for which action $a$ is optimal for a receiver:
\begin{align*}
  \mathcal{P}_a \defeq \left\{ \mu \in \Delta(\Omega) :
  \expec_\mu\left[u(\omega,a)\right] \geq
  \expec_\mu\left[u(\omega,a')\right], ~\text{for all $a' \in A$}\right\}.
\end{align*}
It is without loss of generality to assume that for each $a \in A$,
the set $\mathcal{P}_a$ is non-empty. (This is because a receiver can
never be persuaded to play an action $a \in A$ for which
$\mathcal{P}_a$ is empty, and hence such an action can be dropped from
$A$.)

We consider the following regularity conditions on the persuasion
instances:
\begin{assumption}[Regularity Conditions]
  The instance $\instance$ satisfies the following conditions:
  \begin{enumerate}
  \item \label{as:non-empty-interior} There exists $d > 0$ such that
    for each $a \in A$, the set $\mathcal{P}_a$ contains an
    $\ell_1$-ball of size $d$. Let $D >0$ denote the largest value of
    $d$ for which the preceding is true, and let
    $\eta_a \in \mathcal{P}_a$ be such that
    $\mathsf{B}_1(\eta_a, D) \subseteq \mathcal{P}_a$.
  \item \label{as:prior-set} There exists a $p_0 > 0$ such that for
    all $\mu \in \mathcal{B}_0$ we have
    $\min_{\omega} \mu(\omega) \geq p_0 > 0$.
  \end{enumerate}
\end{assumption}
The first condition requires that each such set $\mathcal{P}_a$ has a
non-empty relative interior; this excludes the pathological instances
like that in Appendix~\ref{ap:regularity-failure}, for which there
exists an action $a$ with $\mathcal{P}_a$ a singleton. We note that
this condition is analogous to the Slater condition in convex
optimization, imposing non-empty interior on the feasibility region to
obtain strong duality. The second condition is technical and is made
primarily to ensure the potency of the first condition: without it,
the sets $\{\mathcal{P}_a\}_{a \in A}$ may satisfy the first condition
in $\Delta(\Omega)$, while failing to satisfy it relative to the
subset $\Delta(\{ \omega : \mu(\omega) > 0\})$ for some
$\mu \in \mathcal{B}_0$. Taken together, these regularity conditions
serve to avoid pathologies, and henceforth we restrict our attention
only to those instances satisfying these regularity conditions.

  Under the regularity conditions, our first result shows that the
  cost of robust persuasion $\gap(\mu, \mathcal{B})$ is at most linear
  in the size of the set $\mathcal{B}$. 
\begin{proposition}\label{prop:upper-bound}
  For any instance that satisfies the regularity conditions, for all
  $\mu \in \mathcal{B}_0$ and for all $\epsilon \geq 0$, we have
  $\gap(\mu, \mathsf{B}_1(\mu, \epsilon)) \leq \left(\frac{4 }{p_0^2
      D}\right) \epsilon$.
\end{proposition}
The proof of the upper bound is obtained through an explicit
construction of a signaling mechanism $\widehat{\sigma}$ that is
persuasive for all distributions in the set
$\mathsf{B}_1(\mu, \epsilon)$, and by showing that the sender's
expected payoff under $\widehat{\sigma}$ is close to that under the
optimal signaling mechanism in $\pers(\mu)$. For this construction, we
first use the geometry of the instance to split the distribution $\mu$
into a convex combination of distributions that either fully reveal
the state, or are well-situated in the interior of the sets
$\mathcal{P}_a$. (It is here that we make use of the two regularity
assumptions.) We then construct the mechanism $\widehat{\sigma}$ to
induce, under prior $\mu$, the aforementioned beliefs as posteriors.
Finally, we show that for any prior $\mu'$ close enough to $\mu$, the
posteriors induced by $\widehat{\sigma}$ are close to the posteriors
induced under prior $\mu$, implying that these posteriors lie within
the sets $\mathcal{P}_a$. This proves the persuasiveness of
$\widehat{\sigma}$ for all distributions $\mu'$ close to $\mu$. We
provide the complete proof in Appendix~\ref{ap:cost-of-rp}.

Next, we provide a (worst-case) lower bound on $\gap$. We accomplish
this by carefully constructing a persuasion instance $\instance_0$
where being robustly persuasive leads to a substantial loss to the
sender. The instance $\instance_0$ has three states
$\Omega = \{ \omega_0, \omega_1, \omega_2\}$ and five actions
$A = \{a_0, a_1, a_2, a_3, a_4 \}$ for the receiver. At a high level,
the receiver's preference can be illustrated as in
Fig.~\ref{fig:belief-regions}, which depicts the receiver's optimal
action for any belief in the simplex. The regions $\mathcal{P}_i$ in
the figure correspond to the set of beliefs that induce action
$a_i \in A$ as the receiver's best response. The instance is crafted
so that the sets $\mathcal{P}_1$ and $\mathcal{P}_2$ that induce
actions $a_1$ and $a_2$ respectively are symmetric and extremely
narrow with the width controlled by an $\ell_1$-ball of radius $D$
contained within, as depicted in Fig.~\ref{fig:regularity-splitting}.
(Since $|\Omega|=3$, the $\ell_1$-ball here is a hexagon.) For
completeness, the receiver's utility is listed explicitly in
Table~\ref{tab:instance}. The sender seeks to persuade the receiver
into choosing one of actions $a_1$ and $a_2$ (regardless of the
state); all other actions are strictly worse for the sender. Formally,
we set $v(\omega, a) = 1$ if $a \in \{a_1, a_2\}$ and $0$ otherwise,
for all $\omega$. The sender's initial knowledge regarding the state
distribution is captured by the set
$\mathcal{B}_0 = \{ \mu \in \Delta(\Omega) : \min_{\omega} \mu \geq
p_0\}$, while the distribution of interest is
$\mu^*=(p_0, \frac{1-p_0}{2},\frac{1-p_0}{2})$, corresponding to the
midpoint of the tips of the sets $\mathcal{P}_i$, as shown in
Fig.~\ref{fig:regularity-splitting}. We focus on the setting where the
instance parameters $D$ and $p_0$ satisfy $Dp_0 < 1/64$. The following
proposition shows that in the instance $\instance_0$, it is costly to
require the signaling mechanism to be robustly persuasive for a set of
distributions around $\mu^*$. The result also implies that the bound
on $\gap(\cdot)$ obtained in Proposition~\ref{prop:upper-bound} is
almost tight, except for a factor of $1/p_0$.
\begin{figure}
  \begin{subfigure}{0.45\textwidth}
    \centering
    \includegraphics[width=0.8\textwidth]{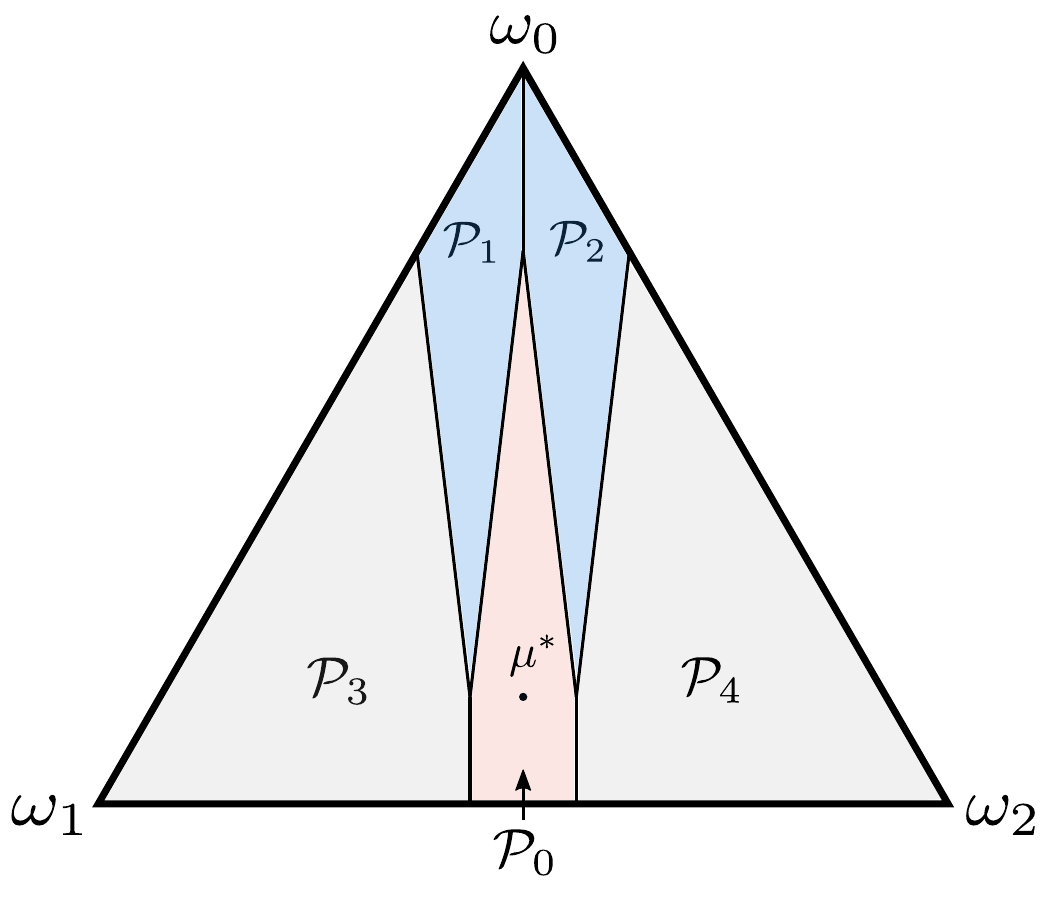}
    \caption{Receiver's preferences}%
    \label{fig:belief-regions}
  \end{subfigure}\hspace{0.05\textwidth}%
  \begin{subfigure}{0.45\textwidth}
    \centering
    \includegraphics[width=0.8\textwidth]{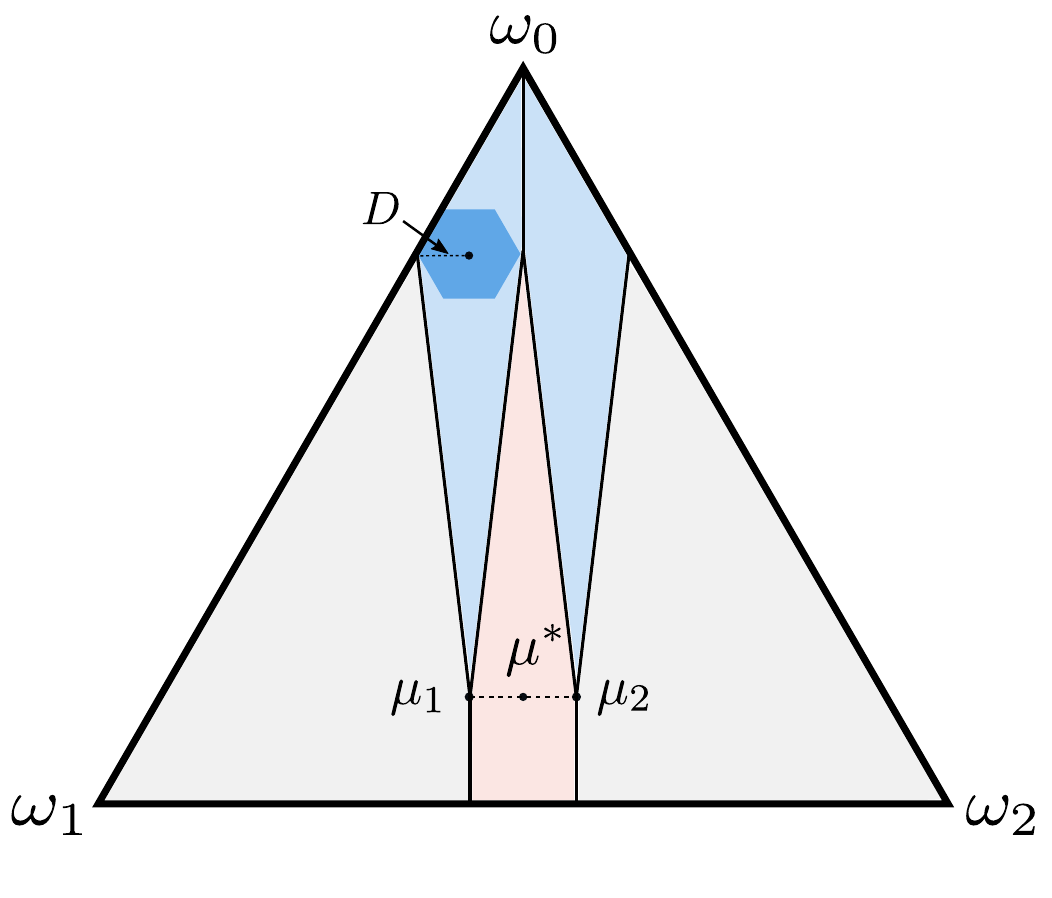}
    \caption{Prior $\mu^*$}%
    \label{fig:regularity-splitting}
  \end{subfigure}
  \caption{The persuasion instance $\instance_0$.}\label{fig:instance}
\end{figure}

\begin{table}
    \TABLE
    {Receiver's utility in instance $\instance_0$, with
      $u(\omega, a_0)$ normalized to $0$ for all $\omega \in\Omega$.\label{tab:instance}}
    {\begin{tabular}{cccccc}
    \hline
    \up & $a_1$ & $a_2$ & $a_3$ & $a_4$ \down\\
    \hline\up
    $\omega_0$ & {\small $2D^2$} & {\small $2D^2$} & {\small $-2D(1-p_0-2D)$} & {\small $-2D(1-p_0-2D)$} \\
    $\omega_1$ & {\small $(1-2D)(1-D)-p_0$} & {\small $(D+1)(2D-1)+p_0$} & {\small $2(1-p_0-2D)(1-D)$} & {\small $-2(1-p_0-2D)(D+1)$} \\
    $\omega_2$ & {\small $(D+1)(2D-1)+p_0$} & {\small $(1-2D)(1-D)-p_0$} & {\small $-2(1-p_0-2D)(D+1)$} & {\small $2(1-p_0-2D)(1-D)$} \down\\
    \hline
  \end{tabular}}
  {}
  \end{table}

\begin{proposition}\label{prop:lower-bound} For the instance
  $\instance_0$, we have $\opt(\mu^*) = 1$. Furthermore, for all
  $\epsilon \in (0, D)$, we have
  \begin{align*}
    \gap(\mu^*, \{\mu^*, \bar{\mu}_1, \bar{\mu}_2\} ) \geq \frac{\epsilon}{8D p_0},
  \end{align*}
  where $\bar{\mu}_1 = \mu^*+ \frac{\epsilon}{2}(e_1-e_2)$,
  $\bar{\mu}_2 = \mu^*+\frac{\epsilon}{2}(e_2-e_1)$, where the belief
  $e_i$ puts all its weight on $\omega_i$.
\end{proposition}

  We defer the rigorous algebraic proof of the lower bound to
  Appendix~\ref{ap:cost-of-rp} and present a brief sketch
  using a geometric argument here. In the instance $\instance_0$, the
 distribution $\mu^*$ can be written as a convex combination
  $\mu^* = (\mu_1+\mu_2)/2$, where $\mu_1$ and $\mu_2$ are the tips of
  regions $\mathcal{P}_1$ and $\mathcal{P}_2$ respectively (see
  Fig.~\ref{fig:regularity-splitting}). Thus, by the splitting
  lemma~\citep{aumann1995repeated}, it follows that the optimal
  signaling mechanism sends signals that induce posterior beliefs
  $\mu_1$ and $\mu_2$ leading to receiver's choice of $a_1$ and $a_2$
  respectively. Since the sender can always persuade the receiver to
  choose one of her preferred actions, we obtain $\opt(\mu^*) =1$. On
  the other hand, for a signaling mechanism to be robustly persuasive
  for all distributions $\epsilon$-close to the distribution $\mu^*$ for
  sufficiently small $\epsilon$, the posteriors for the sender's
  preferred actions $a_1, a_2$ induced by the signaling mechanism have
  to be shifted up significantly in the narrow region. Such a large
  discrepancy ultimately forces the sender to suffer a substantial
  loss in the expected payoff.

  %
  %

%
%
%
%
 %
\section{Regret Analysis}
\label{sec:regret-analysis}

We now return to the regret analysis of the online persuasion setting.
The regret bounds we establish in this section make critical use of
the characterization of the cost of robust persuasion from the
preceding section.

Our main result establishes a upper bound on the regret of the
\leap~algorithm in instances satisfying the regularity conditions.
While $p_0$ appears in our regret bound, it is not required by the
\leap~algorithm for its operation.
\begin{theorem}\label{thm:vanishing-regret} Suppose the instance
  $\instance$ satisfies the regularity condition. For $t \in [T]$, let
  $\epsilon_t = \min\{ \sqrt{\frac{|\Omega|}{t}}(1+ \sqrt{\Phi \log
    T}), 2\}$ with $\Phi > 0$. Then, for all
  $\mu^* \in \mathcal{B}_0$, with probability at least
  $1 - T^{1 - \frac{3 \Phi \sqrt{\Omega}}{56}} - T^{-8 \Phi
    |\Omega|}$, the \leap algorithm satisfies
  \begin{align*}
    \regret_\instance(\leap, \mu^* , T) \leq 2 \left( \frac{20}{p_0^2D}  + 1\right) \left(1 +  \sqrt{|\Omega|T}(1+ 2\sqrt{\Phi \log T})\right).
  \end{align*}
  In particular, the regret is of order
  $O\left( \frac{\sqrt{\Omega}}{p_0^2D}\sqrt{T\log T} \right)$ with
  high probability.
\end{theorem}
The central step in the proof is the following decomposition of the
regret, established in Lemma~\ref{lem:regret-expression} in
Appendix~\ref{ap:vanishing-regret-proof}:
\begin{align*}
  \regret_\instance(\leap, \mu^* , T)
  &\leq \sum_{t \in [T]} \gap(\mu^*, \mathsf{B}_1(\mu^*, {\|\mu^* - \gamma_t\|}_1)) + \sum_{t \in [T]} \gap(\gamma_t, \mathsf{B}_1(\gamma_t, \epsilon_t))\\
  &\quad + \sum_{t \in [T]} {\|\mu^* - \gamma_t\|}_1 +  \sum_{t \in [T]} \left(\expec_{\mu^*}[v(\omega_t, a_t)|h_{t}] - v(\omega_t, a_t)\right).
\end{align*}
Observe that on the event
$\{ \mu^* \in \cap_{t \in [T]} \mathcal{B}_t\}$, we have
${\| \mu^* - \gamma_t\|}_1 \leq \epsilon_t$. Thus, on this event, the
first two terms on the right-hand side of the preceding inequality
capture the cost of persuading robustly for all distributions in an
$\ell_1$-ball of radius $\epsilon_t$ around the distribution $\mu^*$
and its estimate $\gamma_t$. Moreover, the third term represents the
estimation error between $\mu^*$ and $\gamma_t$. Together with
Proposition~\ref{prop:upper-bound}, we thus obtain that the first
three terms are of order
$\sum_{t \in [T]} \epsilon_t = O(\sqrt{T\log T})$. Finally, the last
term, which captures the randomness in the sender's payoff, is also of
the same order due to a simple application of the Azuma-Hoeffding
inequality. The details are provided in
Appendix~\ref{ap:vanishing-regret-proof}.

\subsection{Lower bound}

In this section, we show that our regret upper bound in Theorem
\ref{thm:vanishing-regret} is essentially tight with respect to the
parameters $D,T $ (up to a lower order $\sqrt{\log T}$ factor). We
also show that the inverse polynomial dependence on $p_0$, the
smallest probability of states, is necessary though the exact order of
the dependence on $p_0$ is left as an interesting open question.
\begin{theorem}\label{thm:regret-lower-bound} For the instance
  $\instance_0$ and distribution $\mu^* \in \mathcal{B}_0$ considered in Proposition~\ref{prop:lower-bound}, there exists a $T_0 >0$ such
  that for any $T \geq T_0$ and any $\beta_T$-robustly persuasive algorithm
  $\alg$ the following holds with probability at least
  $\frac{1}{3} - 2\beta_T$:
  \begin{align*}
    \regret_\instance(\alg, T, \mu^*) = T \cdot \opt(\mu^*) - \sum_{t \in [T]} v(\omega_t, a_t) \geq \frac{\sqrt{T}}{32Dp_0}.
\end{align*}
\end{theorem}

We provide a sketch here. First the regret can be split into two
terms:
\begin{align*}
 \regret_\instance(\alg, T, \mu^*) = T \cdot \opt(\mu^*) - \sum_{t \in [T]} V(\mu^*, \sigma^\alg[h_t])+ \sum_{t \in [T]} V(\mu^*, \sigma^\alg[h_t]) - \sum_{t \in [T]} v(\omega_t, a_t) 
\end{align*}

Let $\event(\mu)$ be the event under which the signaling mechanism
  $\sigma^\alg[h_t]$ chosen by the algorithm $\alg$ after any history
  $h_t \in \event(\mu)$ is persuasive for the distribution $\mu$. Hence on the event $\event(\mu^*) \cap \event(\bar{\mu}_1) \cap
  \event(\bar{\mu}_2)$, the signaling mechanism $\sigma^\alg[h_t]$ is persuasive for all three distributions $\mu^*, \bar{\mu}_1 = \mu^* + \frac{\epsilon}{2}(e_1-e_2) $ and $\bar{\mu}_2= \mu^* + \frac{\epsilon}{2}(e_2-e_1) $. From Proposition~\ref{prop:lower-bound}, we have that on this event, the first term, which is the sender's expected loss, is no less than $T \cdot \gap(\mu^*, \{\mu^*, \bar{\mu}_1, \bar{\mu}_2 \})$. We lower bound the second term using the Azuma-Hoeffding inequality. The remaining step is to show that the probability of the event $\event(\mu^*) \cap \event(\bar{\mu}_1) \cap
  \event(\bar{\mu}_2)$ does not vanish as $T$ goes to infinity, which follows from robust persuasiveness of the algorithm $\alg$ and careful choice of $\epsilon$. The details are provided in Appendix~\ref{ap:regret-lower-bound}.

%
%
%
%
%
%
%
%
%
%
%
%
%
%
%
%
%
%
%
%
%

%
%
%
%
%
%
%
%
%
%
%
%
%
%
%
%
%
%
%
%
%
%
%
%
%
%
%
%
%
%
%
%
%

%
%
%
%
 %
\section{Conclusion}
\label{sec:conclusion}

We studied a repeated Bayesian persuasion problem where the
distribution of payoff-relevant states is unknown to the sender. The
sender learns this distribution from observing state realizations
while making recommendations to the receiver.  We propose the \leap
algorithm which persuades robustly and achieves $O(\sqrt{T\log T})$
regret against the optimal signaling mechanism under the knowledge of
the state distribution. To match this upper-bound, we construct a persuasion
instance for which no persuasive algorithm achieves regret better than
$\Omega(\sqrt{T})$. Taken together, our work precisely characterizes
the value of knowing the state distribution in repeated persuasion.

While social learning is a strong motivation for our robust
persuasiveness criterion, there are other motivations as well. For
instance, a platform concerned about its long-run reputation may want
to design a recommendation algorithm that guarantees verifiably good
quality recommendations, not just with respect to currently available
state realization data, but also with respect to any additional data
obtained in the future. An algorithm satisfying our robust
persuasiveness criterion enables such a platform to meet its goals.

While in our analysis we have assumed that the receiver's utility is
fixed across time periods, our model and the analysis can be easily
extended to accommodate heterogeneous receivers, as long as the sender
observes the receiver's type prior to making the recommendation, and
the cost of robustness $\gap$ can be uniformly bounded across
different receiver types. More interesting is the setting where the
sender must persuade a receiver with an unknown type. In such a
setting, assuming the sender cannot elicit the receiver's type prior
to making the recommendation, the sender makes a menu of action
recommendations (one for each receiver type). It can be shown the
complete information problem in this setting corresponds to public
persuasion of a group of receivers with no externality, which is known
to be a computationally hard linear program with exponentially many
constraints~\citep{dughmi2017algorithmic}. Consequently, our algorithm
ceases to be computationally efficient. Nevertheless, our results
imply that the algorithm continues to maintain the $O(\sqrt{T\log T})$
regret bound.

Our characterization of the cost of robust persuasion may be
  of independent interest. For instance, one can derive the sample
  complexity bounds for static persuasion problem when the sender only
  has access to the samples from the underlying distribution. To
  obtain a signaling mechanism that is persuasive with probability at
  least $1-\beta$ and is $\epsilon$-optimal, our characterization
  yields a sample complexity of
  $\Theta(\frac{|\Omega|+ \log (1/\beta)}{p_0^4 D^2 \epsilon^2})$. Note
  that for large enough $\epsilon$, one can simply use the
  full-information mechanism with no need for any samples.

  Our analysis highlights two main technical contributions. One is the
  characterization of the cost of robust persuasion for the underlying
  linear program and using this characterization to perform a tight
  regret analysis for the online learning problem. The former result
  heavily uses the specifics of the persuasion problem (for instance,
  the use of the splitting lemma to construct a feasible robust
  solution) whereas the latter result is more agnostic to the setting.
  Given this, we believe our approach can be extended to other online
  linear programming settings as long as one can obtain a
  characterization of the corresponding cost of robustness. Note that
  even in our persuasion setting, we had to impose the regularity
  conditions to obtain the linear bounds on the cost of robustness,
  without which the cost could be $O(1)$ and the regret would be
  linear. Whether these regularity conditions can be generalized to
  other linear settings is an interesting question for further
  investigation.

%
%
%
%

%
%
%
%
%
%
%
%
%
%

%
%
%
%
%
%

%
%
%
%
%
%
%
%
%
%
%

%
%
%
%

%
%
%

%

%
\begin{APPENDICES}
  %

\section{Examples of Persuasion Instances}
\label{ap:example-instances}

In this section, we provide examples of instances that illuminate
various aspects of our theoretical results.

\subsection{Failure of the Regularity Condition}
\label{ap:regularity-failure}

We begin with an example of an instance $\instance_1$ in which the
regularity condition does not hold, and in which any
$\beta$-robustly persuasive algorithm incurs a linear regret. We establish this
by proving that in this instance, the cost of robust persuasion
$\gap(\mu^*, \mathsf{B}_1(\mu^*, \epsilon))$ is a constant independent
of $\epsilon$ for all $\epsilon>0$.

In the persuasion instance $\instance_1$, the state space is given by
$\Omega=\{ \omega_0, \omega_1, \omega_2\}$ and the receiver has four
actions $A=\{a_0, a_1, a_2, a_3\}$. The receiver's utility is given by
$u(\omega_i,a_j) = \ind\{i=j\} + \frac{1}{3} \ind\{j=3\}$ for
$i \in \{0,1, 2\}$ and $j \in \{0, 1,2,3\}$. The sender's payoff is
given by $v(\omega_i, a_j) = \ind\{j=3\}$; in other words, the sender
strictly prefers the receiver choosing action $a_3$ over any other
action in all states.  The sender's initial knowledge regarding the underlying 
state distribution is captured by $\mathcal{B}_0 = \mathsf{B}_1(\mu^*, \epsilon_0)$
for some $\epsilon_0>0$, where
$\mu^* = (\tfrac{1}{6}, \tfrac{2}{3}, \tfrac{1}{6})$.

The receiver's preferences can be depicted as in
Fig.~\ref{fig:regularity-condition-example}, with sets $\mathcal{P}_j$
for $j \in \{0, 1,2\}$ denoting the set of beliefs for which the
receiver finds it optimal to choose action $a_j$. On the other hand,
the set of beliefs for which it is optimal for the receiver to choose
the sender's preferred action $a_3$ is given by
$\mathcal{P}_3 = \{ \bar{\mu}\}$ where
$\bar{\mu} = (\tfrac{1}{3},\tfrac{1}{3}, \tfrac{1}{3})$ (the orange
central point in the figure). Since $\mathcal{P}_3$ has an empty
interior, the first regularity condition fails for the instance
$\instance_1$.

\begin{figure}[htbp]
\begin{center}
\includegraphics[height=2in]{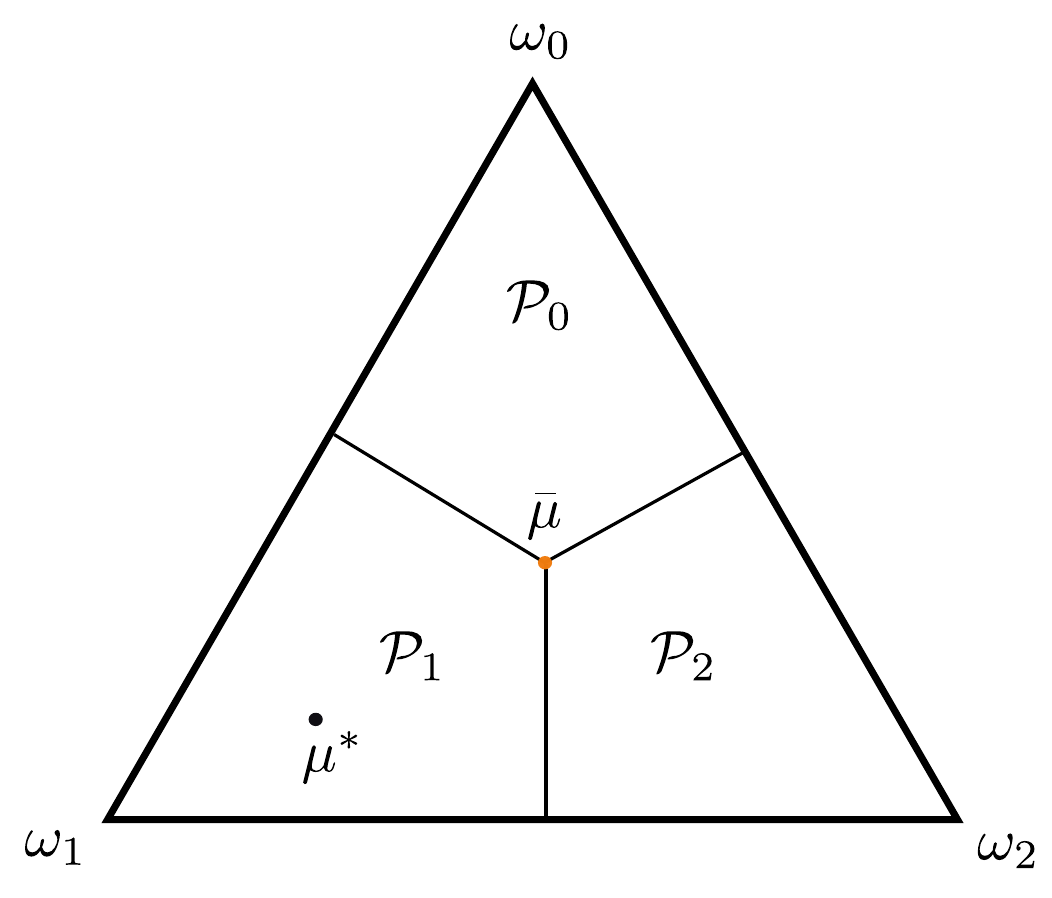}
\caption{\centering{The persuasion instance $\instance_1$.}}
\label{fig:regularity-condition-example}
\end{center}
\end{figure}

If the distribution $\mu^*$ is known, the sender can use a signaling
mechanism that induces $\bar{\mu}$ as the posterior belief with
positive probability, causing the receiver to choose action $a_3$
leading to a positive payoff for the sender. Formally, under the distribution
$\mu^* = (\frac{1}{6}, \frac{2}{3}, \frac{1}{6})$, the optimal signaling mechanism is given by
\begin{align*}
  \sigma^*(\omega_1, a_3) &=1-\sigma^*(\omega_1, a_1) = \frac{1}{4},\\
  \sigma^*(\omega_0, a_3) &= \sigma^*(\omega_2, a_3) = 1,\\
  \sigma^*(\omega,a) &=0, \quad \text{otherwise.}
\end{align*}
Under this mechanism, the sender's utility is given by
$\opt(\mu^*) = \frac{1}{2}$.

However, for any $\epsilon>0$, the only recommendations that are
robustly persuasive for all distributions in $\mathsf{B}_1(\mu^*, \epsilon)$
are $a_0, a_1, a_2$. Thus, any signaling mechanism that is persuasive
for all distributions in $\mathsf{B}_1(\mu^*, \epsilon)$ can never recommend
the sender's preferred action $a_3$, leading to the sender's payoff of
zero. Hence, the difference in the sender's expected utility between
using the optimal persuasive signaling mechanism for the distribution $\mu^*$ and
using the optimal signaling mechanism that is persuasive for all
distributions in $\mathsf{B}_1(\mu^*, \epsilon)$ is given by
\begin{align*}
  \gap(\mu^*, \mathsf{B}_1(\mu^*, \epsilon))  &= V(\mu^*, \sigma^*) - V(\mu^*, \hat{\sigma}) =\frac{1}{2}.
\end{align*}
Thus, $\gap(\mu^*, \mathsf{B}_1(\mu^*, \epsilon))$ is a constant
independent of $\epsilon>0$. Using this bound on the cost of robust
persuasion and an argument similar to the proof of
Theorem~\ref{thm:regret-lower-bound}, one can show that the regret of
any $\beta$-robustly persuasive mechanism is of order $\Omega(T)$ with
probability at least $1/3$.

\subsection{Linear regret for $0$-robustly persuasive mechanisms}
\label{ap:0-persuasive}

In this section, we establish the necessity to consider
$\beta$-robustly persuasive mechanisms with (small) $\beta>0$ for obtaining
meaningful regret bounds. This is demonstrated by a simple example of
the persuasion instance $\instance_2$ in which any $0$-robustly persuasive
algorithm necessarily incurs a linear regret.

In the persuasion instance $\instance_2$, the state space is given by
$\Omega = \{\omega_0, \omega_1 \}$ and the receiver's action space is
given by $A = \{a_0, a_1 \}$. The receiver's utility is given by
$u(\omega_i, a_j) = \ind \{i=j\}$ for $i,j \in \{0,1\}$, i.e., the
receiver desires to ``match'' the action with the state. On the other
hand, the sender strictly prefers the receiver choosing action $a_0$
over action $a_1$ in all states, i.e.,
$v(\omega_i, a_j) = \ind \{j=0\}$ for all $i,j \in \{0,1\}$. The
sender's initial knowledge regarding the distribution is captured by
$\mathcal{B}_0 = \{ (\tfrac{1}{2}-\alpha, \tfrac{1}{2}+\alpha) \colon
\alpha \in [-\tfrac{1}{4},+\tfrac{1}{4}]\}$.

For each $i\in \{0, 1\}$, the set of beliefs for which it is optimal
for the receiver to choose action $a_i$ is given by $\mathcal{P}_i$
where $\mathcal{P}_0 = \{ (a, 1-a) \colon a \in [\tfrac{1}{2}, 1]\}$
and $\mathcal{P}_1 = \{ (a, 1-a) \colon a \in
[0,\tfrac{1}{2}]\}$. Note that the persuasion instance $\instance_2$
satisfies both the regularity conditions.

Now, since all the distributions in $\mathcal{B}_0$ are absolutely continuous
with respect to each other, any algorithm $\alg$ that is
$0$-robustly persuasive must select at each time $t \in [T]$ a signaling
mechanism $\sigma_t$ in the set $\pers(\mathcal{B}_0)$. However, it is
straightforward to verify that among all mechanisms that are persuasive
for all distributions in $\mathcal{B}_0$, the one that maximizes sender's
payoff is given by
$\hat{\sigma}(\omega_0, a_0) = 1- \hat{\sigma}(\omega_0, a_1) = 1,
\hat{\sigma}(\omega_1, a_0) = 1- \hat{\sigma}(\omega_1, a_1) =
\frac{1}{3}$. For the distribution
$\mu^* = (\tfrac{1}{2}, \tfrac{1}{2}) \in \mathcal{P}_0 \cap
\mathcal{B}_0$, it follows that the sender's payoff under
$\hat{\sigma}$ is $V(\mu^*, \hat{\sigma}) = \frac{2}{3}$.

On the other hand, since $\mu^* \in \mathcal{P}_0 \cap \mathcal{B}_0$,
the signaling mechanism that recommends action $a_0$ in both states is
persuasive for $\mu^*$, and thus achieves an expected payoff of
$\opt(\mu^*) = 1$. Thus, we deduce that for the distribution $\mu^*$, any
$0$-robustly persuasive algorithm must incur a constant regret of at least
$\tfrac{1}{3}$ at each time leading to an overall regret linear in
$T$.

\section{Proofs from Section~\ref{sec:leap}}
\label{ap:persuasiveness-proof}

This section provides the proof of Theorem~\ref{thm:persuasiveness},
along with a helper lemma establishing the concentration of the
empirical distribution around the (unknown) distribution.

\proof{Proof of Theorem~\ref{thm:persuasiveness}.}
  If $\mu^* \in \mathcal{B}_{t}$ for each $t \in [T]$, then since
  $\sigma[h_{t}]$ is persuasive under all distributions in $\mathcal{B}_{t}$,
  we deduce that $\sigma[h_{t}]$ is persuasive under the distribution $\mu^*$ for
  all $t\in [T]$. Thus, we obtain that the \leap-algorithm is
  $\beta$-robustly persuasive for
  \begin{align*}
    \beta = \sup_{\mu^* \in \mathcal{B}_0} \prob_{\mu^*}\left( \cap_{t
    \in [T]} \mathcal{B}_{t} \not\ni \mu^*\right).
  \end{align*}
  Now, for any $\mu \in \mathcal{B}_0$, using the union bound we get
  \begin{align*}
    \prob_{\mu}\left( \cap_{t\in [T]} \mathcal{B}_{t} \not\ni \mu\right)
    &= \prob_{\mu}\left( \cup_{t\in [T]} \mathcal{B}_{t}^c \ni \mu\right)\\
    &\leq \sum_{t \in [T]} \prob_{\mu}\left( \mathcal{B}_{t}^c \ni \mu\right)\\
    &= \sum_{t\in[T]} \prob_{\mu}\left( {\| \gamma_t- \mu\|}_1 > \epsilon_t\right)\\
    &= \sum_{t\in[T]} \prob_{\mu}\left( {\| \gamma_t- \mu\|}_1 > \sqrt{\frac{|\Omega|}{t}} \left(1 + \sqrt{\Phi \log T}\right) \right).
  \end{align*}
  For $t < \frac{1}{4}\Phi \log T$, we have
\begin{align*}
  \sqrt{\frac{|\Omega|}{t}} \left(1 + \sqrt{\Phi \log T}\right) > 2 \sqrt{|\Omega|} \left(1 +  \frac{1}{\sqrt{\Phi \log T} }\right) \geq 2.
\end{align*}
Hence,
$\prob_{\mu}\left( {\| \gamma_t- \mu\|}_1 > \sqrt{\frac{|\Omega|}{t}}
  \left(1 + \sqrt{\Phi \log T}\right) \right) = 0$. On the other hand,
for $t \geq \frac{1}{4}\Phi \log T$, we have
$\sqrt{\Phi \log T} \leq 2\sqrt{t}$, and hence from
Lemma~\ref{lem:better-bound-on-norm}, we obtain
\begin{align*}
  \sum_{t \geq \frac{1}{4}\Phi \log T} \prob_{\mu}\left( {\| \gamma_t- \mu\|}_1 > \sqrt{\frac{|\Omega|}{t}} \left(1 + \sqrt{\Phi \log T}\right) \right)
  &\leq \sum_{t \geq \frac{1}{4}\Phi \log T} \exp\left( - \frac{3 \Phi \log T \sqrt{\Omega}}{56}\right)\\
  &\leq  T^{-\frac{3\Phi\sqrt{\Omega}}{56}} \left( T - \frac{\Phi \log T}{4}\right)\\
  &\leq T^{1 - \frac{3\Phi\sqrt{\Omega}}{56}}.
\end{align*}
Setting $\Phi > 20$ implies that the final term is at most
$T^{-0.5}$.~\Halmos\endproof

The following lemma provides a bound on the $\ell_1$-norm of the
deviation of the empirical distribution from its mean.

\begin{lemma}\label{lem:better-bound-on-norm} For each $t \in [T]$, and
  for any $\mu \in \Delta(\Omega)$, we have for all
  $0 < \Phi_t \leq 2\sqrt{t}$,
    \begin{align*}
      \prob_{\mu} \left({\| \gamma_t - \mu\|}_1 \geq \sqrt{\frac{|\Omega|}{t}} \left(1 + \Phi_t\right)\right)
      &\leq \exp\left( - \frac{3 \Phi_t^2 \sqrt{\Omega}}{56}\right) \ind\left\{\sqrt{\frac{|\Omega|}{t}} \left(1 + \Phi_t\right) \leq 2 \right\}.
    \end{align*}
  \end{lemma}
\proof{Proof.}
  Let $X_t \in \{0, 1\}^{|\Omega|}$ denote the random variable with
  $X_t(\omega) = \ind\{\omega_t = \omega\}$, and define
  $Y_t = X_t - \expec_\mu[X_t]$.  Let
  $Z_t = {\| \sum_{\tau \in [t]} Y_\tau\|}_1$. Since
  ${\|Y_t\|}_1 \leq {\| X_t - \expec_\mu[ X_t] \|}_1 \leq 2$ for each
  $t \in [T]$, by \citet[Corollary 8.46]{foucartH13}, we obtain for
  each $t \in [T]$,
  \begin{align*}
    \prob_{\mu} \left(Z_t \geq \expec_\mu[Z_t] + s\right) \leq \exp\left(- \frac{3s^2}{4\left(6t + 6 \expec_{\mu}[Z_t] + s\right)} \right).
  \end{align*}
  Next, letting $Z_{t,\omega} = | \sum_{\tau \in [t]} Y_\tau(\omega)|$
  for $\omega \in \Omega$, we obtain
  \begin{align*}
    \expec_{\mu}[Z_t]
    &= \sum_{\omega \in \Omega} \expec_\mu [ Z_{t, \omega}]\\
    &= \sum_{\omega \in \Omega} \expec_\mu [ \sqrt{Z_{t, \omega}^2}]\\
    &\leq \sum_{\omega \in \Omega} \sqrt{\expec_\mu [ Z_{t, \omega}^2]}\\
    &= \sum_{\omega \in \Omega} \sqrt{\sum_{\tau \in [t]} \mathbf{Var}_\mu[Y_\tau(\omega)]}\\
    &= \sqrt{t} \cdot \sum_{\omega \in \Omega} \sqrt{\mu(\omega) (1 - \mu(\omega))}\\
    &\leq \sqrt{|\Omega| t},
  \end{align*}
  where the first inequality follows from Jensen's inequality, and the
  third equality follows from the fact that, since
  $\expec_\mu[Y_t(\omega)] =0$, we have
  $\expec[Z_{t, \omega}^2] =\sum_{\tau \in [t]}
  \mathbf{Var}_\mu[Y_\tau(\omega)]$. The final step follows from a
  straightforward optimization. Thus, we obtain
  \begin{align*}
    \prob_{\mu} \left(Z_t \geq \sqrt{|\Omega| t} + s\right) \leq \exp\left(- \frac{3s^2}{4\left(6t + 6\sqrt{|\Omega| t} + s\right)} \right).
  \end{align*}
  Choosing $s = \Phi_t \sqrt{|\Omega|t}$ for
  $0 < \Phi_t\leq 2\sqrt{t}$, and noting that
  $Z_t = t {\| \gamma_t - \mu\|}_1$, we obtain
  \begin{align*}
    \prob_{\mu} \left({\| \gamma_t - \mu\|}_1 \geq \sqrt{\frac{|\Omega|}{t}} \left(1 + \Phi_t\right) \right)
    &\leq \exp\left(- \frac{3\Phi_t^2|\Omega|t}{4\left(6t + 6\sqrt{|\Omega| t} +  \Phi_t \sqrt{|\Omega|t}\right)} \right)\\
    &\leq \exp\left(- \frac{3\Phi_t^2\sqrt{|\Omega|}}{4\left(12 +\Phi_t/\sqrt{t}\right)} \right)\\
    &\leq \exp\left( - \frac{3 \Phi_t^2 \sqrt{\Omega}}{56}\right).
  \end{align*}
  The lemma statement then follows after noticing that for all
  $t\in [T]$, we have
  ${\|\gamma_t - \mu\|}_1 \leq {\|\gamma_t\|}_1 + {\|\mu\|}_1 \leq
  2$.%
~\Halmos\endproof

\section{Proofs from Section~\ref{sec:cost-of-rp}}
\label{ap:cost-of-rp}

This section provides the proofs of the propositions in
Section~\ref{sec:cost-of-rp}. Throughout, we use the same notation as
in the main text.
\proof{Proof of Proposition~\ref{prop:upper-bound}.} Observe that for
  $\epsilon > \frac{p_0^2 D}{4}$, we have
  $ \frac{4\epsilon }{p_0^2 D} > 1$, and hence the specified bound is
  trivial. Hence, hereafter, we assume
  $\epsilon \leq \frac{p_0^2 D}{4}$.

  To begin, let
  $\sigma \in \arg\max_{\sigma' \in \pers(\mu)} V(\mu, \sigma')$
  denote the optimal signaling mechanism under the distribution $\mu$. Let
  $A_{+} = \{ a \in A : \sum_{\omega \in \Omega} \sigma(\omega, a) >
  0\}$ denote the set of all actions that are recommended with
  positive probability under $\sigma$.  For each $a \in A_{+}$, let
  $\mu_a$ denote the receiver's posterior belief (under signaling
  mechanism $\sigma$) upon receiving the action recommendation
  $a$. Note that since $\sigma$ is persuasive under $\mu$, we must
  have $\mu_a \in \mathcal{P}_a$. By the splitting lemma~\citep{aumann1995repeated}, it
  then follows that $\mu$ can be written as a convex combination
  $\sum_{a \in A_{+}} w_a \mu_a$ of $\{\mu_a : a \in A_{+}\}$, where
  $w_a \in [0,1]$ is given by
  $w_a= \sum_{\omega \in \Omega} \mu(\omega) \sigma(\omega, a)$.

  We next explicitly construct a signaling mechanism
  $\widehat{\sigma}$. To simplify the proof argument, the signaling
  mechanism $\widehat{\sigma}$ we construct is not a {\em
    straightforward} mechanism, in the sense that it reveals more than
  just action recommendations for signals in $S$. Using revelation
  principle, one can construct an equivalent straightforward mechanism
  $\bar{\sigma}$ by
  \emph{coalescing}~\citep{anunrojwong2020persuading} signals with the
  same best response for the signal. We omit the details of this
  reduction. We start with some definitions that are needed to
  construct the signaling mechanism $\widehat{\sigma}$.

  Let $\eta_a \in \mathcal{P}_a$ be such that
  $\mathsf{B}_1(\eta_a, D) \subseteq \mathcal{P}_a$. For
  $\delta = \frac{2\epsilon}{p_0 D} \in [0,1]$, define
  $\xi_a = (1- \delta) \mu_a + \delta \eta_a \in \mathcal{P}_a$ for
  each $a \in A_{+}$ and let $\xi = \sum_{a \in A_{+}} w_a \xi_a$.
  Furthermore, since $\mu_a \in \mathcal{P}_a$ and
  $\mathsf{B}_1(\eta_a, D) \subseteq \mathcal{P}_a$, the convexity of
  the set $\mathcal{P}_a$ implies that
  $\mathsf{B}_1(\xi_a, \delta D) \subseteq \mathcal{P}_a$.

  Since $\mu \in \mathcal{B}_0 \subseteq \relint(\Delta(\Omega))$, we
  have $\frac{1}{1-\rho}(\mu - \rho \xi) \in \Delta(\Omega)$ for all
  small enough $\rho > 0$.  Let
  $\bar{\rho} \defeq \sup\left\{ \rho \in [0, 1] : \frac{1}{1-\rho}
    (\mu - \rho \xi) \in \Delta(\Omega)\right\}$ be the largest such
  value in $[0,1]$, and define $\chi$ as
  \begin{align*}
    \chi \defeq \begin{cases} \frac{1}{1-\bar{\rho}} \left(\mu - \bar{\rho}\xi\right), & \text{if
        $\bar{\rho} < 1$;}\\
      \mu, & \text{if $\bar{\rho} =1$.}
    \end{cases}
  \end{align*}
  Then, we obtain
  $\mu = \bar{\rho} \xi + (1- \bar{\rho})
  \chi$. Furthermore, if $\bar{\rho} < 1$, we have
 \begin{align*}
   \bar{\rho}  =  \frac{{\|\chi - \mu\|}_1}{{\|\chi - \mu\|}_1 + {\| \mu - \xi\|}_1 } \geq \frac{p_0}{p_0 +  \delta},
  \end{align*}
  where the inequality follows from
  ${\| \mu - \xi\|}_1 \leq \sum_{a \in A_{+}} w_a {\| \mu_a -
    \xi_a\|}_1 = \delta \sum_{a \in A_{+}} w_a {\| \eta_a - \xi_a\|}_1
  \leq 2\delta$ and from the fact that $\chi$ lies in the boundary of
  $\Delta(\Omega)$, which implies
  ${\| \chi - \mu\|}_1 \geq 2 \min_{\omega} \mu(\omega) \geq 2 p_0$.

  With the preceding definitions in place, we are now ready to
  construct the mechanism $\widehat{\sigma}$. Let $a_\omega$ be a best
  response for the receiver at state $\omega \in \Omega$, and let
  $S = \{ (\omega, a_\omega) \in \Omega \times A : \chi(\omega) >
  0\}$. Consider the signaling mechanism $\widehat{\sigma}$, with the
  set of signals $A_{+} \cup S$, defined as follows: for each
  $\omega \in \Omega$, let
  \begin{align}\label{eq:constructed-sigma}
    \widehat{\sigma}(\omega, s)
    &\defeq \begin{cases}
      \bar{\rho} \frac{w_a \xi_a(\omega)}{\mu(\omega)}, &
      \text{for $s=a \in A_{+}$;}\\
      (1-\bar{\rho}) \frac{\chi(\omega)}{\mu(\omega)}, & \text{for
        $s = (\omega, a_\omega) \in S$;}\\
      0, & \text{otherwise.}
    \end{cases}
  \end{align}

  We now show that the signaling mechanism $\widehat{\sigma}$ is
  persuasive for all distributions in $\mathsf{B}_1(\mu, \epsilon)$, in the
  sense that for all signals $s \in A_{+}$ it is optimal for the
  receiver to play $s$, and for all signals
  $s = (\omega, a_\omega) \in S$, it is optimal for the receiver to
  play $a_\omega$.  To see this, for any
  $\gamma \in \mathsf{B}_1(\mu, \epsilon)$, let $\gamma(\cdot |s)$
  denote the receiver's posterior under signaling mechanism
  $\widehat{\sigma}$ upon receiving the signal $s \in A_{+} \cup
  S$. For $s = (\omega, a_\omega) \in S$, we have
  $\gamma( \cdot |s) = e_\omega$, where $e_\omega$ is the belief that
  puts all its weight on $\omega \in \Omega$. Thus, upon receiving the
  signal $s = (\omega, a_\omega)$ it is optimal for the receiver with the
  distribution $\gamma$ to take action $a_\omega$. Thus, it only remains to
  show that signals $s = a \in A_{+}$ are persuasive.

  For $a \in A_{+}$, we have for $\omega \in \Omega$,
  \begin{align*}
    \mu(\omega| a) &= \frac{\mu(\omega) \widehat{\sigma}(\omega, a)}{\sum_{\omega' \in \Omega }\mu(\omega')\widehat{\sigma}(\omega', a)} = \xi_a(\omega)\\
    \gamma(\omega| a) &= \frac{\gamma(\omega) \widehat{\sigma}(\omega, a)}{\sum_{\omega' \in \Omega }\gamma(\omega')\widehat{\sigma}(\omega', a)} =   \frac{\gamma(\omega) }{\mu(\omega)} \cdot \frac{\xi_a(\omega)}{ \sum_{\omega' \in \Omega} \frac{\gamma(\omega') \xi_a(\omega')}{\mu(\omega')}}.
  \end{align*}

  Then, using triangle inequality and some algebra, we obtain
  \begin{align*}
    \left\| \gamma( \cdot | a) - \mu( \cdot | a) \right\|_1
    &= \sum_{\omega \in \Omega} \left|  \gamma(\omega|a) - \xi_a(\omega)\right|\\
    &\leq \sum_{\omega \in \Omega} \left|   \gamma(\omega|a) -  \frac{\gamma(\omega) }{\mu(\omega)} \cdot  \xi_a(\omega)\right| + \sum_{\omega \in \Omega} \left| \frac{\gamma(\omega) }{\mu(\omega)} \cdot  \xi_a(\omega) - \xi_a(\omega)\right|\\
    &\leq  2 \cdot \sup_{\omega \in \Omega} \frac{\xi_a(\omega) }{\mu(\omega)} \cdot {\| \gamma - \mu\|}_1\\
    &\leq \frac{2\epsilon}{p_0},
  \end{align*}
  where in the final inequality, we have used
  $\min_\omega \mu(\omega) \geq p_0$ to get
  $\sup_{\omega \in \Omega} \frac{\xi_a(\omega) }{\mu(\omega)} \leq
  \frac{1}{p_0}$. Since $\mu(\cdot| a) = \xi_a$, this implies that
  $\gamma( \cdot |a) \in \mathsf{B}_1\left(\xi_a,
    \frac{2\epsilon}{p_0}\right) = \mathsf{B}_1(\xi_a, \delta D)
  \subseteq \mathcal{P}_a$. Thus, the signal $a \in A_{+}$ is
  persuasive for the distribution $\gamma \in \mathsf{B}_1(\mu,
  \epsilon)$. Taken together, we obtain that the signaling mechanism
  $\widehat{\sigma}$ is persuasive for all
  $\gamma \in \mathsf{B}_1(\mu, \epsilon)$.

  The persuasiveness of $\widehat{\sigma}$ for all
  $\gamma \in \mathsf{B}_1(\mu, \epsilon)$ implies that
  \begin{align*}
    \sup_{\sigma' \in \pers(\mathsf{B}_1(\mu, \epsilon))} V(\mu, \sigma')
    &\geq V(\mu, \widehat{\sigma})\\
    &=  \sum_{\omega \in \Omega} \sum_{a \in A_{+}} \mu(\omega) \widehat{\sigma}(\omega,a) v(\omega, a) + \sum_{\omega \in \Omega} \sum_{s \in S} \mu(\omega) \widehat{\sigma}(\omega,s) v(\omega, a_\omega) \\
    &\geq \sum_{\omega \in \Omega} \sum_{a \in A_{+}}  \bar{\rho} w_a \xi_a(\omega)  v(\omega, a)\\
    &=  \bar{\rho}  \sum_{\omega \in \Omega} \sum_{a \in A_{+}} w_a \left( (1-  \delta)\mu_a(\omega) + \delta \eta_a(\omega) \right)  v(\omega, a)\\
    &\geq \bar{\rho} (1-  \delta)  \sum_{\omega \in \Omega} \sum_{a \in A_{+}} w_a \mu_a(\omega) v(\omega, a) \\
    &= \bar{\rho} (1-  \delta) \opt(\mu).
  \end{align*}
  Thus, we obtain
  \begin{align*}
    \gap(\mu, \mathsf{B}_1(\mu, \epsilon))
    &= \opt(\mu)  - \sup_{\sigma' \in \pers(\mathsf{B}_1(\mu, \epsilon))} V(\mu, \sigma')\\
    &\leq \left(1 - \bar{\rho}(1 - \delta)\right) \opt(\mu)\\
    &\leq \left(\frac{4}{p_0^2 D}\right)\epsilon ,
  \end{align*}
  where the final inequality follows from
  $\bar{\rho} \geq \frac{p_0}{p_0 + \delta}$,
  $\delta = \frac{2\epsilon}{p_0D}$ and
  $\opt(\mu) \leq 1$.~\Halmos\endproof

\proof{Proof of Proposition~\ref{prop:lower-bound}.}
  It is straightforward to verify that the following signaling
  mechanism $\sigma^* \in \pers(\mu^*)$ optimizes the sender's
  expected utility among all mechanisms in $\pers(\mu^*)$:
  \begin{align*}
    \sigma^*(\omega_0, a_1) &= \sigma^*(\omega_0, a_2) = \frac{1}{2}, \\
    \sigma^*(\omega_1, a_1) &= \sigma^*(\omega_2, a_2) = \frac{1}{2}+\frac{D}{2(1-p_0)}, \\
    \sigma^*(\omega_1, a_2) &= \sigma^*(\omega_2, a_1) = \frac{1}{2}-\frac{D}{2(1-p_0)}, \\
    \sigma^*(\omega, a) &= 0, \quad \text{otherwise}.
  \end{align*}
  Since the action recommendations are always in $\{ a_1, a_2\}$, we
  obtain $\opt(\mu^*) =1$.

  Recall that
  $\bar{\mu}_1 = \mu^* + \frac{\epsilon}{2}(e_1-e_2), \bar{\mu}_2 =
  \mu^* + \frac{\epsilon}{2}(e_2-e_1)$. By the linearity of obedience
  constraints and $\mu^* = (\bar{\mu}_1 + \bar{\mu}_2)/2$, it follows
  that $\pers(\{ \mu^*, \bar{\mu}_1, \bar{\mu}_2 \})$ can be obtained
  by imposing the obedience constraints at distributions $\bar{\mu}_1$
  and $\bar{\mu}_2$. The optimization problem
  $\max_\sigma \{ V(\mu^*, \sigma) : \sigma \in \pers(\{\bar{\mu}_1,
  \bar{\mu}_2 \})\}$ can be solved to obtain the following optimal
  signaling mechanism:
  \begin{align*}
    \hat{\sigma}(\omega_0, a_1) &= \hat{\sigma}(\omega_0, a_2) = \frac{1}{2},\\
    \hat{\sigma}(\omega_1, a_1) &= \hat{\sigma}(\omega_2, a_2) = \frac{X}{Z},\\
    \hat{\sigma}(\omega_1, a_2) &= \hat{\sigma}(\omega_2, a_1) = \frac{Y}{Z},\\
    \hat{\sigma}(\omega_1, a_3) &= \hat{\sigma}(\omega_2, a_4) = 1-\hat{\sigma}(\omega_1, a_1)- \hat{\sigma}(\omega_1, a_2), \\
    \hat{\sigma}(\omega, a) &= 0, \quad \text{otherwise},
  \end{align*}
  where
  \begin{align*}
    X &= 2p_0(1-p_0-\epsilon)(1-p_0+D)D^2+p_0(1-p_0+\epsilon)(1-p_0-D-2D^2),\\
    Y &= p_0(1-p_0-\epsilon)(1-p_0-3D+2D^2)+2p_0(1-p_0+\epsilon)(1-p_0-D)(1-2D)D^2,\\
    Z &= (1-p_0+\epsilon)^2(1-p_0-D)(1-2D)(1-p_0-D-2D^2)\\
      &\quad -(1-p_0-\epsilon)^2(1-p_0+D)(1-p_0-3D+2D^2).
  \end{align*}
  The difference in the sender's expected utility between using the
  optimal persuasive signaling mechanism for the distribution
  $\mu^* \in \mathcal{B}$ and using the optimal signaling mechanism that
  is persuasive for all distributions in $\{\mu^*, \bar{\mu}_1, \bar{\mu}_2 \}$ is given by

  \begin{align*}
    \gap(\mu^*, \pers(\mu^*, \bar{\mu}_1, \bar{\mu}_2)) & = V(\mu^*, \sigma^*) - V(\mu^*, \hat{\sigma}) \\
                                        & \geq  \frac{\epsilon}{2}\frac{1/2+D p_0 (1+\epsilon/2-D p_0 -D)}{Dp_0+\epsilon} \\
                                        & \geq \frac{\epsilon}{8D p_0}.\Halmos
  \end{align*}
\endproof

\section{Proofs from Section~\ref{sec:regret-analysis}}
\subsection{Proof of Theorem~\ref{thm:vanishing-regret}}
\label{ap:vanishing-regret-proof}
In this section, we provide the proof of
Theorem~\ref{thm:vanishing-regret}. In the process, we also state and
prove several helper lemmas used in the proof.

\proof{Proof of Theorem~\ref{thm:vanishing-regret}.} 
In Lemma~\ref{lem:regret-expression}, we obtain the
following bound on the regret:
\begin{align*}
    \regret_\instance(\leap, \mu^* , T)
    &\leq  \sum_{t \in [T]} \gap( \mu^*, \mathsf{B}_1(\mu^*,
      {\| \mu^* - \gamma_t\|}_1))  + \sum_{t \in [T]} \gap(\gamma_{t}, \mathcal{B}_{t})\\
    &\quad + \sum_{t \in [T]} {\| \mu^* -
      \gamma_t\|}_1 + \sum_{t \in [T]} \left(\expec_{\mu^*}[v(\omega_t, a_t)|h_{t}] - v(\omega_t, a_t)\right).
\end{align*}
Now, from Proposition~\ref{prop:upper-bound}, we have
\begin{align*}
  \gap( \mu^*, \mathsf{B}_1(\mu^*, {\| \mu^* - \gamma_t\|}_1))
  &\leq   \left(\frac{4}{p_0^2D}\right) \cdot {\|\mu^* - \gamma_t\|}_1.
\end{align*}
Thus, we obtain
\begin{align*}
    \regret_\instance(\leap, \mu^* , T)
    &\leq  \sum_{t \in [T]} \gap(\gamma_{t}, \mathcal{B}_{t}) +  \left( \frac{4}{p_0^2D}  + 1\right) \sum_{t \in [T]} {\|\mu^* - \gamma_t\|}_1  \\
    &\quad   + \sum_{t \in [T]} \expec_{\mu^*}[v(\omega_t, a_t)|h_{t}] - v(\omega_t, a_t).
  \end{align*}
  Finally, in Lemma~\ref{lem:gap-gamma}, we show that on the event
  $\{\mu^* \in \mathcal{B}_t\}$, we have
  $\gap(\gamma_t, \mathcal{B}_t) \leq \left(\frac{16}{p_0^2
      D}\right) \epsilon_t$. Thus, on the event $\{\mu^* \in
  \cap_{t\in [T]}\mathcal{B}_t\}$, we obtain
  \begin{align*}
    \regret_\instance(\leap, \mu^* , T)
    &\leq    \left( \frac{20}{p_0^2D}  + 1\right) \sum_{t \in [T]} \epsilon_t  + \sum_{t \in [T]} \expec_{\mu^*}[v(\omega_t, a_t)|h_{t}] - v(\omega_t, a_t)\\
    &\leq    \left( \frac{20}{p_0^2D}  + 1\right) \left(2 + \sum_{t =1}^{T-1} \sqrt{\frac{|\Omega|}{t}}(1+ \sqrt{\Phi \log T})\right)\\
    &\quad + \sum_{t \in [T]} \expec_{\mu^*}[v(\omega_t, a_t)|h_{t}] - v(\omega_t, a_t)\\
    &\leq   2  \left( \frac{20}{p_0^2D}  + 1\right) \left(1 +  \sqrt{|\Omega|T}(1+ \sqrt{\Phi \log T})\right)\\
    &\quad + \sum_{t \in [T]} \expec_{\mu^*}[v(\omega_t, a_t)|h_{t}] - v(\omega_t, a_t),
  \end{align*}

  where in the final inequality, we have used the fact that
  $\sum_{t=1}^{T-1} 1/\sqrt{t} \leq 2 \sqrt{T}$.

  From Theorem~\ref{thm:persuasiveness}, we have
  $ \prob_{\mu}\left( \cap_{t \in [T]} \mathcal{B}_{t} \not\ni
    \mu\right) \leq T^{1 - \frac{3 \Phi \sqrt{\Omega}}{56}}$.
  For $t \in [T]$, let
  $X_t \defeq \expec_{\mu^*}[ v(\omega_t, a_t)|h_t] - v(\omega_t,
  a_t)$. Observe that $\expec_{\mu^*}[X_t|h_t] = 0$ and
  $|X_t| \leq 1$. Thus the sequence
  $\{ X_t : t \in [T]\}$ is a bounded martingale difference
  sequence. Hence, from Azuma-Hoeffding~\citep{boucheron2013concentration}, we obtain for
  $z \geq 0$,
  \begin{align*}
    \prob_{\mu^*}\left( \sum_{t \in [T]} \expec_{\mu^*}[ v(\omega_t, a_t)|h_t] - v(\omega_t, a_t) \geq z \right) < \exp\left(-\frac{2z^2}{T}\right).
  \end{align*}
  Choosing $z= \sqrt{ \alpha T\log T}$ with $\alpha > 0$, we have %
  \begin{align*}
      \prob_{\mu^*}\left( \sum_{t \in [T]} \expec_{\mu^*}[ v(\omega_t, a_t)|h_t] - v(\omega_t, a_t) \geq \sqrt{ \alpha T\log T} \right) < \frac{1}{T^{2\alpha}}.
    \end{align*}
    After choosing $\alpha = 4\Phi |\Omega|$ and taking the union
    bound, we obtain with probability at least $1 - T^{1 - \frac{3
        \Phi \sqrt{\Omega}}{56}} - T^{-8  \Phi |\Omega|}$, 
    \begin{align*}
      \regret_\instance(\leap, \mu^* , T) &\leq  2  \left( \frac{20}{p_0^2D}  + 1\right) \left(1 +  \sqrt{|\Omega|T}(1+ 2\sqrt{\Phi \log T})\right).\Halmos
    \end{align*}
\endproof

\begin{lemma}\label{lem:regret-expression} The \leap algorithm
  satisfies
  \begin{align*}
    \regret_\instance(\leap, \mu^* , T)
    &\leq  \sum_{t \in [T]} \gap( \mu^*, \mathsf{B}_1(\mu^*,
      {\| \mu^* - \gamma_t\|}_1))  + \sum_{t \in [T]} \gap(\gamma_{t}, \mathcal{B}_{t})\\
    &\quad + \sum_{t \in [T]} {\| \mu^* -
      \gamma_t\|}_1 + \sum_{t \in [T]} \left(\expec_{\mu^*}[v(\omega_t, a_t)|h_{t}] - v(\omega_t, a_t)\right).
  \end{align*}
\end{lemma}
\proof{Proof.}
  From the definition~\eqref{eq:regret} of regret, we have
  \begin{align}
    \regret_\instance(\leap, \mu^*, T)
    &= \opt(\mu^*)\cdot T - \sum_{t \in [T]} v(\omega_t, a_t) \notag\\
    &= \opt(\mu^*)\cdot T - \sum_{t \in [T]} \expec_{\mu^*}[v(\omega_t, a_t)|h_{t}] + \sum_{t \in [T]} \left(\expec_{\mu^*}[v(\omega_t, a_t)|h_{t}] - v(\omega_t, a_t)\right)\notag\\
    &=   \sum_{t \in [T]}\left( \opt(\mu^*) - V(\mu^*, \sigma[h_{t}])\right) + \sum_{t \in [T]} \left(\expec_{\mu^*}[v(\omega_t, a_t)|h_{t}] - v(\omega_t, a_t)\right),\label{eq:regret-expression}
  \end{align}
  where in the last equality, we have used the fact that
  $\expec_{\mu^*}[v(\omega_t, a_t)|h_{t}] = V(\mu^*, \sigma[h_t])$.
  Moreover, note that
  \begin{align*}
    \opt(\mu^*) - V(\mu^*, \sigma[h_t])
    &= \opt(\mu^*) - V(\gamma_t, \sigma[h_t]) + V(\gamma_t, \sigma[h_t])  - V(\mu^*, \sigma[h_t])\\
    &= \left(\opt(\mu^*) - \opt(\gamma_t)\right) +  \left( \opt(\gamma_t) - V(\gamma_t, \sigma[h_t])\right)\notag \\
    &\quad + \left(V(\gamma_t, \sigma[h_t])  - V(\mu^*, \sigma[h_t])\right)\\
    &= \left(\opt(\mu^*) - \opt(\gamma_t)\right) +  \gap(\gamma_t, \mathcal{B}_t)\\
    &\quad + \left(V(\gamma_t, \sigma[h_t])  - V(\mu^*, \sigma[h_t])\right),
  \end{align*}
  where in the final equality, we have used the fact that
  $\opt(\gamma_t) - V(\gamma_t, \sigma[h_t]) = \gap(\gamma_t,
  \mathcal{B}_t)$. Substituting the preceding expression into
  \eqref{eq:regret-expression} yields
  \begin{align*}
    \regret_\instance(\leap, \mu^* , T)
    &=   \sum_{t \in [T]}\left( \opt(\mu^*) - \opt(\gamma_{t})\right) + \sum_{t \in [T]} \gap(\gamma_{t}, \mathcal{B}_{t})\\
    &\quad + \sum_{t \in [T]} \left( V(\gamma_{t}, \sigma[h_{t}]) - V(\mu^*, \sigma[h_{t}])\right)\\
    &\quad + \sum_{t \in [T]} \left(\expec_{\mu^*}[v(\omega_t, a_t)|h_{t}] - v(\omega_t, a_t)\right).
  \end{align*}
  Now, in Lemma~\ref{lem:opt-lipschitz}, we prove
  $\opt(\mu^*) -\opt(\gamma_t) \leq \gap( \mu^*, \mathsf{B}_1(\mu^*,
  {\| \mu^* - \gamma_t\|}_1)) + \frac{1}{2} \cdot {\| \mu^* -
    \gamma_t\|}_1$. Furthermore, in Lemma~\ref{lem:val-unif-cont}, we
  show that
  $V(\gamma_{t}, \sigma[h_{t}]) - V(\mu^*, \sigma[h_{t}]) \leq
  \frac{1}{2} {\| \mu^* - \gamma_t \|}_1$. Putting it all together
  yields the lemma statement.~\Halmos\endproof

\begin{lemma}\label{lem:opt-lipschitz} For any $\mu_1,\mu_2 \in
  \Delta(\Omega)$, we have
  \begin{align*}
    \opt(\mu_1) - \opt(\mu_2) &\leq \gap(\mu_1, \mathsf{B}_1(\mu_1,{\|\mu_1 - \mu_2\|}_1))  + \frac{1}{2}\cdot{\left\|\mu_1 - \mu_2\right\|}_1.
  \end{align*}
\end{lemma}
\proof{Proof.} Fix $\mu_1, \mu_2 \in \Delta(\Omega)$. For
  $i \in \{1, 2\}$, let
  $\sigma_i \in \arg\max_{\sigma' \in \pers(\mu_i)} V(\mu_i,
  \sigma')$. By definition, we have
  $\opt(\mu_i) = V(\mu_i, \sigma_i)$.

  Next, among all signaling mechanisms that are persuasive for all
  $\mu \in \mathsf{B}_1(\mu_1, {\|\mu_1 - \mu_2\|}_1)$, let $\sigma_3$
  maximize $V(\mu_1, \sigma)$. Since $\sigma_3$ is persuasive for
  $\mu_2$, we have
  $\opt(\mu_2) = V(\mu_2, \sigma_2) \geq V(\mu_2, \sigma_3)$. Thus, we
  have
  \begin{align*}
    \opt(\mu_1) - \opt(\mu_2)
    &= V(\mu_1, \sigma_1) - V(\mu_2, \sigma_2)\\
    &\leq V(\mu_1, \sigma_1) - V(\mu_2, \sigma_3)\\
    &= V(\mu_1, \sigma_1) - V(\mu_1, \sigma_3) + V(\mu_1, \sigma_3) - V(\mu_2, \sigma_3)\\
    &\leq \gap(\mu_1, \mathsf{B}_1(\mu_1,{\|\mu_1 - \mu_2\|}_1))  + \frac{1}{2}\cdot  {\left\|\mu_1 - \mu_2\right\|}_1.
  \end{align*}
  Here, the inequality follows from the definition of $\gap(\cdot)$,
  and from Lemma~\ref{lem:val-unif-cont}.~\Halmos\endproof

\begin{lemma}\label{lem:val-unif-cont}
  For any $\mu_1,\mu_2 \in \Delta(\Omega)$ and any signaling mechanism
  $\sigma$, we have
  \begin{align*}
    \left|V(\mu_1, \sigma) - V(\mu_2, \sigma)\right| &\leq \frac{1}{2} \cdot {\| \mu_1 - \mu_2\|}_1.
  \end{align*}
\end{lemma}
\proof{Proof.} Fix $\mu_1, \mu_2 \in \Delta(\Omega)$. For any signaling
  mechanism $\sigma$ that is persuasive under $\mu_1$, we have for any
  $x \in \reals$,
  \begin{align*}
    |V(\mu_1, \sigma) - V(\mu_2, \sigma)|
    &= \left|\sum_{\omega \in \Omega} \left(\mu_1(\omega) - \mu_2(\omega)\right) \left(\sum_{a \in A} \sigma(\omega, a) v(\omega, a) - x\right)\right|\\
    &\leq {\| \mu_1 - \mu_2\|}_1 \cdot \sup_{\omega \in \Omega} \left|\sum_{a \in A} \sigma(\omega, a) v(\omega, a) - x\right|,
  \end{align*}
  where we have used the H\"older's inequality in the last line.
  Optimizing over $x$, together with the fact that the sender's
  valuations lie in $[0,1]$, yields the result.~\Halmos\endproof

\begin{lemma}\label{lem:gap-gamma} For $t\in [T]$, on the event
  $\{\mu^* \in \mathcal{B}_t\}$, we have
  \begin{align*}
    \gap(\gamma_t, \mathcal{B}_t)
    &      \leq \left(\frac{16}{p_0^2 D}\right) \epsilon_t.
  \end{align*}
\end{lemma}
\proof{Proof.}
On the event $\{\mu^* \in \mathcal{B}_t\}$, we have
  \begin{align*}
    \gamma_t(\omega)
    &\geq \mu(\omega) - {\| \gamma_t - \mu^*\|}_1\\
    &\geq p_0 - \epsilon_t.
  \end{align*}
  Thus, for $\epsilon_t < \frac{p_0}{2}$, we have
  $\min_\omega \gamma_t(\omega) \geq \frac{p_0}{2}$.  Using the same
  argument as in Proposition~\ref{prop:upper-bound}, we then obtain
  \begin{align*}
    \gap(\gamma_t, \mathcal{B}_t)
    &= \gap(\gamma_t, \mathsf{B}_1(\gamma_t, \epsilon_t))
      \leq \left(\frac{4}{D \min_\omega \gamma_t(\omega)^2}\right) \epsilon_t
      \leq \left(\frac{16}{p_0^2 D}\right) \epsilon_t.
  \end{align*}
  For $\epsilon_t > p_0/2$, the bound holds trivially since
  $16\epsilon_t/p_0^2D > 1$.~\Halmos\endproof

\subsection{Proof of Theorem~\ref{thm:regret-lower-bound}}
\label{ap:regret-lower-bound}
We conclude this section with the proof of the lower bound in
Theorem~\ref{thm:regret-lower-bound}.

\proof{Proof of Theorem~\ref{thm:regret-lower-bound}.} For a distribution
$\mu \in \mathcal{B}_0$, define the event $\event(\mu)$ as
  \begin{align*}
    \event(\mu) = \left\{ h_T : \sigma^\alg[h_t] \in \pers(\mu),
    ~\text{for each $t \in [T]$}\right\}.
  \end{align*}
  In words, under the event $\event(\mu)$, the signaling mechanism
  $\sigma^\alg[h_t]$ chosen by the algorithm $\alg$ after any history
  $h_t \in \event(\mu)$ is persuasive for the distribution $\mu$. Since the
  algorithm $\alg$ is $\beta_T$-robustly persuasive, we obtain
  \begin{align*}
    \prob_\mu\left( \event(\mu)\right) \geq 1 -\beta_T, \quad \text{for
    all $\mu \in \mathcal{B}_0$.}
  \end{align*}

  Fix an $\epsilon \in (0,\tfrac{1 - 3p_0}{2})$ to be chosen later,
  and consider the distributions
  $\mu^* = (p_0, \frac{1-p_0}{2},\frac{1-p_0}{2})$ and
  $\bar{\mu}_1 = \mu^* + \frac{\epsilon}{2} \left(e_1 - e_2\right)$
  and
  $\bar{\mu}_2 = \mu^* + \frac{\epsilon}{2} \left(e_2 - e_1\right)$,
  where $e_j$ is the belief that puts all its weight on state
  $\omega_j$ for $j \in \{1,2\}$. Observe that
  $\prob_{\mu^*}\left( \event(\mu^*)\right) \geq 1 -\beta_T$ since
  $\mu^* \in \mathcal{B}_0$. Note that for each $j \in \{1, 2\}$ and
  for all $\epsilon \in (0, \tfrac{1-3p_0}{2})$, we have
  $\bar{\mu}_j \in \mathcal{B}_0$ and hence
  $\prob_{\bar{\mu}_j}\left( \event(\bar{\mu}_j)\right) \geq 1
  -\beta_T$.

  Now, on the event
  $\event(\mu^*) \cap \event(\bar{\mu}_1) \cap
  \event(\bar{\mu}_2)$, the signaling mechanism $\sigma^\alg[h_t]$
  chosen by the algorithm after any history $h_t$ is persuasive for
  all the distributions $\mu^*, \bar{\mu}_1, \bar{\mu}_2$. Thus on the event
  $\event(\mu^*) \cap \event(\bar{\mu}_1) \cap
  \event(\bar{\mu}_2)$, we have
  \begin{align*}
    T \cdot \opt(\mu^*) - \sum_{t \in [T]} V(\mu^*, \sigma^\alg[h_t])
    &\geq  T \cdot  \gap(\mu^*, \{\mu^*, \bar{\mu}_1, \bar{\mu}_2\}) \geq \frac{\epsilon T}{8D p_0},
  \end{align*}
  where the first inequality follows from the definition of $\gap$
  in~\eqref{eq:gap}, and the second inequality follows from
  Proposition~\ref{prop:lower-bound}.

  Now, we have
  \begin{align*}
    2 \left| \prob_{\mu^*}\left( \event(\bar{\mu}_1)\right) - \prob_{\bar{\mu}_1}\left( \event(\bar{\mu}_1)\right)\right|^2
    &\leq \sum_{t \in [T]} \kl{\mu^*}{\bar{\mu}_1}\\
    &= \frac{1-p_0}{2} \log\left( \frac{(1-p_0)^2}{(1-p_0)^2 - \epsilon^2} \right)T\\
    &= \frac{1-p_0}{2} \log\left( 1 + \frac{\epsilon^2}{(1-p_0)^2 - \epsilon^2} \right)T\\
    &\leq \frac{1-p_0}{2} \left( \frac{\epsilon^2}{(1-p_0)^2 - \epsilon^2} \right)T,
  \end{align*}
  where the first inequality is the Pinsker's inequality, and the
  first equality is from the definition of the Kullback-Leibler
  divergence, and the final inequality follows from
  $\log( 1 + x) \leq x$ for $x \geq 0$. Thus, for
  $\epsilon < \tfrac{1-p_0}{2}$, we obtain
  \begin{align*}
    2 \left| \prob_{\mu^*}\left( \event(\bar{\mu}_1)\right) - \prob_{\bar{\mu}_1}\left( \event(\bar{\mu}_1)\right)\right|^2 \leq \frac{2\epsilon^2T}{3(1-p_0)} \leq \epsilon^2T,
  \end{align*}
  where we have used $p_0 \leq \frac{1}{|\Omega|} = \frac{1}{3}$ in
  the final inequality. Thus, we obtain that
  \begin{align*}
    \prob_{\mu^*}\left(\event(\bar{\mu}_1)\right)
    &\geq \prob_{\mu^*}\left(\event(\mu^*)\right)  - \left| \prob_{\mu^*}\left( \event(\bar{\mu}_1)\right) - \prob_{\bar{\mu}_1}\left( \event(\bar{\mu}_1)\right)\right|\\
    &\geq 1 - \beta_T -  \epsilon\sqrt{\frac{T}{2}}.
  \end{align*}
  By the same argument, we obtain
  $\prob_{\mu^*}\left(\event(\bar{\mu}_2)\right) \geq 1 - \beta_T -
  \epsilon\sqrt{\frac{T}{2}}$.

  By the linearity of the obedience constraints, we obtain that if
  $\sigma \in \pers(\bar{\mu}_1) \cap \pers(\bar{\mu}_2)$, then
  $\sigma \in \pers(\mu^*)$. Thus, we have
  $\event(\bar{\mu}_1) \cap \event(\bar{\mu}_2) \subseteq \event(\mu^*)$, and
  hence
  \begin{align*}
    \prob_{\mu^*}( \event(\mu^*) \cap \event(\bar{\mu}_1) \cap \event(\bar{\mu}_2))
    &= \prob_{\mu^*}(\event(\bar{\mu}_1) \cap \event(\bar{\mu}_2))\\
    &\geq  \prob_{\mu^*}(\event(\bar{\mu}_1)) + \prob_{\mu^*}(\event(\bar{\mu}_2)) -1\\
    &\geq 1 - 2 \beta_T - \epsilon\sqrt{2T}.
  \end{align*}
  Finally, by the Azuma-Hoeffding inequality, we obtain
  \begin{align*}
    \prob_{\mu^*} \left( \sum_{t \in [T]} V(\mu^*, \sigma^\alg[h_t]) - \sum_{t \in [T]} v(\omega_t, a_t) < - \sqrt{T} \right) < e^{-1/2}.
  \end{align*}

  Taken together, we obtain that with probability at least
  $1 - 2\beta_T - \epsilon\sqrt{2T} - e^{-1/2}$, we have
  \begin{align*}
    \regret_\instance(\alg, T, \mu^*) = T \cdot \opt(\mu^*) - \sum_{t \in [T]} v(\omega_t, a_t) \geq  \frac{\epsilon T}{8D p_0} - \sqrt{T}.
  \end{align*}
  For $T \geq T_0 = \frac{1}{(1-3p_0)^2}$, choosing
  $\epsilon = \frac{1}{32\sqrt{T}} \leq \frac{1-3p_0}{2}$, we obtain,
  with probability at least $\frac{1}{3} - 2\beta_T$,
\begin{align*}
  \regret_\instance(\alg, T, \mu^*) = T \cdot \opt(\mu^*) - \sum_{t \in [T]} v(\omega_t, a_t) \geq  \sqrt{T}\left(\frac{1}{16D p_0} - 1 \right) \geq \frac{\sqrt{T}}{32Dp_0},
\end{align*}
for $Dp_0 < 1/32$.~\Halmos\endproof
%
%

%
%
%
%
%
%
%
%
%
%
%
%
%
%
%
%
%
%
%
%
%
%
%
%

%
%
%
%
%
%

%
%
%
%
%
%
%
%
%
%
%
%
%
%
%
%
%

%
%
%
%
%
%
%
%
%
%
%
%
%
%
%
%
%
%
%
%
%
%
%
%
%
%

%
%
%
%
%
%
%
%
%
%
%
%
%
%
%
%
%
%

%
%
%
%
%
%
%
%
%
%
%
%
%

%
%
%
%
%
%
%

%
%
%
%
 \end{APPENDICES}

  \ACKNOWLEDGMENT{The first and the second author gratefully
    acknowledge partial support from the National Science Foundation
    under grant CMMI-2002156. The third author is supported by an NSF
    award CCF-2303372, an ARO award W911NF-23-1-0030 and a Google
    Faculty Research Award. A preliminary version of this work
    appeared as an extended abstract at the 22nd ACM Conference on
    Economics and Computation (EC 2021).}

\bibliographystyle{informs2014}
\bibliography{learning-persuade}

\end{document}